\documentclass[a4paper,UKenglish]{article}

\usepackage{amssymb,amsmath, amsthm}
\usepackage{calrsfs} 

\usepackage{authblk}

\newtheorem{theorem}{Theorem} 
\newtheorem{prop}[theorem]{Proposition}
\newtheorem{proposition}[theorem]{Proposition}
\newtheorem{lemma}[theorem]{Lemma}

\theoremstyle{definition}
\newtheorem{definition}[theorem]{Definition}

\theoremstyle{remark}

\theoremstyle{definition}

\newcommand{\CC}{\mathcal C}
\newcommand{\DD}{\mathcal D}
\newcommand{\EE}{\mathcal E}
\newcommand{\OO}{\mathcal O}
\newcommand{\OOl}{\OO_{\preceq}}

\newcommand{\NN}{\mathbb N}

\newcommand{\vCP}{\vdash_{{}_{\! \! \mathrm {CP}}}}
\newcommand{\PV}{\mathrm{PV}}
\newcommand{\ER}{\mathrm{ER}}
\newcommand{\DRAT}{\mathrm{DRAT}}
\newcommand{\CP}{\mathrm{CP}}
\newcommand{\ERpls}{\mathrm{ER\textrm{-}PLS}}

\newcommand{\ang}[1]{\langle #1 \rangle}

\newcommand{\llbracket}{[ \hspace{-1.4pt} [}
\newcommand{\rrbracket}{ ] \hspace{-1.4pt}  ]}
\newcommand{\norm}[1]{\llbracket #1 \rrbracket}

\newcommand{\lelex}{\le_{\mathrm{lex}}}
\newcommand{\llex}{<_{\mathrm{lex}}}
\newcommand{\Lit}{\mathrm{Lit}}
\newcommand{\Sat}{\mathrm{Sat}}
\newcommand{\Refu}{\mathrm{Ref}}
\newcommand{\out}{\mathrm{out}}

\def\restrict{ \upharpoonright  }
\let\epsilon\varepsilon
\let\phi\varphi

\title{The strength of the dominance rule} 

\author{
Leszek Aleksander Kołodziejczyk\thanks{Institute of Mathematics, University of Warsaw, Poland, \texttt{lak@mimuw.edu.pl}}  
~and Neil Thapen\thanks{Institute of Mathematics, Czech Academy of Sciences, Czechia, \texttt{thapen@math.cas.cz}. Supported by the Institute of Mathematics, Czech Academy of Sciences (RVO 67985840) and GA\v{C}R grant 23-04825S.}
}

\begin{document}

\maketitle

\begin{abstract} 
It has become standard that,
when a SAT solver decides that a CNF $\Gamma$ is unsatisfiable, 
it produces a certificate of unsatisfiability in the form of a refutation of $\Gamma$ in some  proof system.
The system typically used is DRAT,
which is equivalent to extended resolution (ER)
-- for example, until this year DRAT refutations were required 
in the annual SAT competition.
Recently [Bogaerts et al.~2023] introduced a new proof system, associated with the tool VeriPB,
which is at least as strong as DRAT 
and is further able to handle certain symmetry-breaking techniques. 
We show that this system simulates the proof system $G_1$, 
which allows limited reasoning with QBFs
and forms the first level above ER in a natural hierarchy
of proof systems. This hierarchy is not known to be strict,
but nevertheless this  is evidence that the system of [Bogaerts et al. 2023]  is plausibly
strictly stronger than ER and DRAT.
In the other direction, we show that symmetry-breaking for a single symmetry 
can be handled inside ER.
\end{abstract}

\section{Introduction}

We write $\Lit$ for the set of propositional literals and $0, 1$ for the propositional constants.
A \emph{clause} is a disjunction which may contain $0$, $1$ or a literal
together with its negation;
in the last two cases we call it \emph{tautologous}.
 A \emph{conjunctive normal form} formula, or CNF,
is a set of clauses, understood as a conjunction.
We  write $\bot$ for the empty clause and $\top$ for the empty CNF.
For a clause $C = x_1 \vee \dots \vee x_k$ we write $\neg C$ for 
the CNF $\neg x_1 \wedge \dots \wedge \neg x_k$.

A \emph{substitution} is a map 
$\omega : \Lit \cup \{0,1\} \rightarrow \Lit \cup \{0,1\}$
which respects negations and is the identity on~$\{0,1\}$.
We view a {partial assignment} as a  kind of substitution.
For literals $p_1, \dots, p_k$ we write ${\vec p}_{\restrict \omega}$ 
for the tuple $\omega(p_1), \dots, \omega(p_k)$.
For a clause $C$, we write $C_{\restrict \omega}$ for the clause $\{\omega(p):p \in C \}$.
For a CNF $\Gamma$, we write $\Gamma_{\restrict \omega}$
for $\{ C_{\restrict \omega} :  C \in \Gamma \}$. 
We write $\omega \vDash \Gamma$ to mean that $\Gamma_{\restrict \omega}$ is tautologous,
that is, every clause is tautologous.
If $\omega$ is a partial assignment
this is the same as the usual meaning of $\vDash$.
The composition of substitutions $\tau, \omega$ is defined by
$\tau \circ \omega (p) = \tau(\omega(p))$.
Note that $\Gamma_{\restrict \tau \circ \omega} = (\Gamma_{\restrict \omega})_{\restrict \tau}$
and thus $\tau \vDash \Gamma_{\restrict \omega}$ if and only if $\tau \circ \omega \vDash \Gamma$.

\begin{definition}
A \emph{symmetry} of a CNF $\Gamma$ is a substitution $\omega$
such that $\Gamma_{\restrict \omega} = \Gamma$.
\end{definition}

Hard combinatorial formulas often have many symmetries. 
A successful heuristic to make such formulas easier for SAT solvers is
\emph{symmetry-breaking}, usually in the form of introducing
a  \emph{lex-leader} constraint~\cite{crawford1996symmetry}. 
We illustrate this in the following proposition.
Suppose $\Gamma$ is a CNF in variables~$z_1, \dots, z_n$
and suppose we have a formula $[x_1, \dots, x_n \lelex y_1, \dots, y_n]$ expressing the lexicographic
order on assignments to $\vec x$ and $\vec y$. For now we suppress technical 
issues of how exactly we express this, and tacitly treat $[\vec x \lelex \vec y]$ as though it is a CNF.

\begin{proposition} \label{pro:intro_symmetry}
If $\omega$ is a symmetry of $\Gamma$, then
$\Gamma$ and $\Gamma \cup [\vec z \lelex \vec z_{\restrict \omega}]$
are equisatisfiable.
\end{proposition}

\begin{proof}
Suppose $\Gamma$ is satisfiable. Let $\alpha$ be a
lexicographically minimal assignment such that $\alpha \vDash \Gamma$.
Then by symmetry  $\alpha \vDash \Gamma_{\restrict \omega}$
and thus $\alpha \circ \omega \vDash \Gamma$.
By minimality we have $\alpha \lelex \alpha \circ \omega$
which, working through the definitions, gives us that $\alpha \vDash [\vec z \lelex \vec z_{\restrict \omega}]$
as required.
\end{proof}

The new formula $\Gamma \cup [\vec z \lelex \vec z_{\restrict \omega}]$
is potentially much easier to solve than the original $\Gamma$, because
the extra constraint can substantially shrink the space of partial assignments
that the solver has to search through. Note that the proof 
still works if we simultaneously add lex-leader constraints for several different symmetries.

It is now common for SAT solvers to have a \emph{proof-logging} component
which, when it decides that a formula $\Gamma$ is unsatisfiable, will generate a certificate
of unsatisfiability -- in other words, a refutation of $\Gamma$. 
For this to be useful, it must be in a well-known proof system for which 
trusted software exists to verify that the refutation is correct. 
The standard system used for this is DRAT~\cite{wetzler2014drat}.
However, it is open to what extent DRAT can (feasibly) handle
reasoning that uses  symmetry-breaking,
especially for more than one symmetry~\cite{heule2015expressing}.

This issue is addressed in~\cite{bgmn:journal}
 which introduces
a new proof-logging system with tools to handle quite general symmetry-breaking,
extending DRAT.
It builds on the system VeriPB~\cite{elffers2020justifying} and as such uses reasoning
using linear inequalities, rather than clauses; furthermore it is equipped to solve optimization problems, rather
than just  satisfiability.

We study the proof complexity of this system (as a refutation system for CNFs).
We show that it is equivalent to~$G_1$, 
a system based on limited reasoning with QBFs which is
one level above extended resolution (ER) 
in a natural hierarchy of proof systems~\cite{kp:quantified-propositional}. 
This is in contrast to the redundancy-based systems which have been studied
recently, such as DRAT, propagation redundancy and substitution redundancy,
which are  in their full generality all equivalent to 
ER~\cite{kiesl2018extended, heule2020strong, buss2021drat}.
This may represent another step in the strength of
reasoning used in SAT algorithms,
like from DPLL and treelike resolution, 
to CDCL and resolution, to the just-mentioned systems and ER~\cite{DLL62, ClauseLearning, beame2004towards}.
In particular, it is unlikely that the new system  in~\cite{bgmn:journal}
can be simulated by DRAT.

The main tool used in~\cite{bgmn:journal} for symmetry-breaking is called the 
\emph{dominance-based strengthening rule}.
It is based on the following principle, which we express here in the language of CNFs
and call ``informal'' because we are sweeping under the carpet the exact nature
of the formulas $\Delta$ and 
$[\vec x \llex \vec y]$ (intended to express strict lexicographic order).

\begin{proposition}[informal] \label{pro:informal_rule}
Let $\Gamma, \Delta$ be formulas, where $\Gamma$ is in variables $\vec z$ and has the property that
any assignment satisfying $\Gamma$ can be extended to satisfy $\Delta$. Suppose
we have a clause $C$ and a substitution $\omega$ such that
\begin{equation} \label{eq:rule_intro}
\Gamma \wedge \Delta \wedge \neg C 
\vDash 
\Gamma_{\restrict \omega} \wedge [ \vec z_{\restrict \omega} \llex \vec z ].
\end{equation}
Then $\Gamma$ and $\Gamma \wedge C$ are equisatisfiable.
\end{proposition}

\begin{proof}[Proof (sketch)]
Suppose $\Gamma$ is satisfiable. Let $\alpha$ be a
lexicographically minimal assignment 
such that $\alpha \vDash \Gamma$. Extend $\alpha$ to $\alpha \cup \beta$ satisfying
 $\Gamma \wedge \Delta$.
If $\alpha \cup \beta \vDash C$ then we are done. Otherwise, from the 
entailment~(\ref{eq:rule_intro}) we know
$\alpha \cup \beta \vDash \Gamma_{\restrict \omega} \wedge [ \vec z_{\restrict \omega} \llex \vec z ]$.
Letting $\alpha' = (\alpha \cup \beta) \circ \omega$, we conclude that $\alpha' \vDash \Gamma$
and $\alpha' \llex \alpha$, contradicting the minimality of $\alpha$.
\end{proof}

The rule says roughly: if we have derived $\Gamma$, and have available a suitable $\Delta$, $\omega$
and a proof of entailment~(\ref{eq:rule_intro}), then we can derive $\Gamma \wedge C$.
(In fact, the full rule in~\cite{bgmn:journal} is  more general, since it is not
restricted to the lexicographic ordering.)

Machinery exists to study the strength of rules of this kind by studying how easy it is to prove their soundness
in the first-order setting of bounded arithmetic. Specifically, 
to carry out the proof of
Proposition~\ref{pro:informal_rule} it is sufficient to know that
a nonempty polynomial-time set  (in this case the set of $\alpha$ such that~$\alpha \vDash \Gamma$)
always has a lexicographically least element. This puts the system inside the theory~$T^1_2$,
which is associated with $G_1$ (see below for definitions).

The harder direction, lower-bounding the strength of the rule, uses similar machinery. 
Consider the transformation $\alpha \mapsto \alpha'$ in the proof above.
We claim that, given an arbitrary polynomial time function~$f$, we can construct an instance
of the rule where this transformation is given by~$f$. We do this by defining $\Delta$ to compute~$f$ and
store the resulting values in~$\beta$, 
and defining $\omega$ to pull  these values from $\beta$ back to the $\vec z$ variables,
as $\alpha'$. Observe that in the proof, it is not actually necessary
for $\alpha$ to be lexicographically minimal; the ``local minimality'' property that $\alpha \lelex f(\alpha)$
is sufficient. We show that a converse holds: that using the rule, we can find 
a local minimum of this kind, or to say it more precisely,  we can efficiently derive a contradiction
from the statement that there is no local minimum.

The problem of finding such a local minimum is known as polynomial local search,
or PLS~\cite{jpy:localsearch}. 
It is known that if we are only interested in sentences of low quantifier complexity, 
such as ``CNFs in this family are unsatisfiable'', then
every logical consequence  of ``every polynomial-time set has a least element'',
that is, of $T^1_2$, is already a consequence of the apparently weaker statement
 ``every PLS problem has a solution''.
Using this we can show that, roughly speaking, every CNF which can
be proved unsatisfiable in $T^1_2$ has a short refutation in the proof system in~\cite{bgmn:journal}, 
where the refutation
uses an instance of the rule
constructed from a PLS problem as described above.
It follows that the system simulates~$G_1$.

\medskip

The rest of the paper fills in the details of these arguments.
In Section~\ref{sec:trad} we recall the definitions of some proof systems we will
need. These are $G_1$ itself; cutting planes, which is the foundation for the 
system in~\cite{bgmn:journal}; and $\ER$,
where we will need to work extensively with derivations as well as
 refutations.
In Section~\ref{sec:dominance} we define our version of the system in~\cite{bgmn:journal},
where we have removed the machinery for handling optimization problems.
We call this the \emph{dominance proof system}, and in it
the dominance-based strengthening rule can be used for rather general orderings,
not just lexicographic.
We will work more with a restriction of it,
the \emph{linear dominance proof system}, in which it is limited
to essentially the lexicographic ordering.
We also define a simpler auxiliary proof system $\ERpls$, which 
uses clauses rather than inequalities and captures the properties of the
system that are important for us (illustrated in Proposition~\ref{pro:informal_rule}). 
It  will follow from our results that, as a system for refuting CNFs, it is equivalent to the linear
dominance system.
In Section~\ref{sec:BA} we describe some results we need from bounded arithmetic and give a formal definition
of PLS.
We then show our main result,
\begin{theorem}
The linear dominance system is equivalent to $G_1$.
\end{theorem}

The proof is in three parts.
Section~\ref{sec:ERpls_simulates_G1} contains the main technical work of the paper, 
showing, in Theorem \ref{the:ERpls_simulates_G1}, that $\ERpls$ simulates~$G_1$, as sketched above. 
Section~\ref{sec:lin-dom>er-pls} shows that the linear dominance system
in turn simulates $\ERpls$.
For the remaining direction, that $G_1$ simulates linear dominance,
it is enough to show that linear dominance is sound, provably in~$T^1_2$. 
This is in Section~\ref{sec:soundness},
where we also briefly discuss the difference between 
linear dominance and the full dominance system.

In Section~\ref{sec:upper} we study what we \emph{can} prove in ER about fragments of
these systems. 
We show essentially that, in Proposition~\ref{pro:informal_rule}
above, if the mysterious formula $\Delta$ is not present then we do not need
to use minimality, or even PLS, in the proof. This is because without the step where
we extend~$\alpha$ to satisfy~$\Delta$, 
the move from~$\alpha$ to $\alpha'$ does not involve any computation, but amounts to
shuffling the components of $\alpha$ around using the substitution $\omega$,
and for simple graph-theoretical reasons we can compute the $i$th iteration $\omega^i$
in polynomial time without invoking any stronger principles. 
In Section~\ref{subsec:weak-dominance} we use this observation to show 
a technical result, that a 
natural weakening of the linear dominance system is already simulated by ER.
In Section~\ref{sec:ER_breaking} we use a similar construction to show%
\footnote{We believe this is more general than the result about single symmetries
in~\cite{heule2015expressing}, since we handle an arbitrary symmetry,
not just an involution; see the discussion in~\cite{bgmn:journal}.}:

\begin{theorem} \label{the:intro_Q1}
Define the system $Q_1$ as $\ER$ plus the power to add a lex-leader constraint for a single symmetry.
Then  $Q_1$ is simulated by $\ER$ (and thus by $\DRAT$).
\end{theorem}

Appendix~\ref{sec:technical_appendix} contains some technical
material postponed from other sections.

\medskip

We finish this section by addressing the natural question: what is this
hierarchy of proof systems above $\ER$, and why should we expect it to be strict?
After all, $\ER$ is already a very strong system with many tools for proving combinatorial
and algebraic
statements, and seems to lie well beyond  current methods for proving 
lower bounds~\cite{razborov2021survey}.

$\ER$ was shown in~\cite{cook:feasibly} to correspond to the theory $\PV$,
which models reasoning with polynomial-time concepts 
(see Section~\ref{sec:BA} for definitions).
In ~\cite{kp:quantified-propositional} the systems $G_0, G_1, G_2, \dots$
of quantified Boolean reasoning 
were introduced, to correspond in the same way to 
a hierarchy $T^0_2, T^1_2, T^2_2, \dots$ of bounded arithmetic theories extending $\PV$
(which we can identify with~$T^0_2$~\cite{jerabek:sharply}), where 
$T^1_2$ can reason with $\mathrm{P^{NP}}$ concepts, $T^2_2$ with $\mathrm{P^{\Sigma^p_2}}$, and so on.
In particular, if we ignore issues of uniformity, the unsatisfiable
CNFs with short refutations in $G_i$ capture precisely the 
universal sentences%
\footnote{That is, sentences consisting of a sequence of unbounded universal quantifiers
followed by a polynomial time predicate. We could also write $\forall \Pi^b_1$.}
 provable in~$T^i_2$.
It is a classical result that the fragments $\mathrm{I}\Sigma_0, \mathrm{I}\Sigma_1, \dots$
of Peano arithmetic are separated by universal sentences. 
Specifically, the consistency statement for $\mathrm{I}\Sigma_k$ has this form and is provable
in $\mathrm{I}\Sigma_{k+1}$ but not in $\mathrm{I}\Sigma_k$ (see e.g.~\cite[Chapter I.4(c)]{hajek-pudlak}).
It is expected, essentially by analogy, that the analogous theories $T^0_2, T^1_2, \dots$
are also separated at the universal level by some kind of consistency statement,
although it is known that classical consistency will not work~\cite{wp:boundedinduction}.

In the case we are interested in here, of $\PV$ and $T^1_2$, there is some evidence
of separation at the $\forall \Sigma^b_1$ level (one step above universal) 
since it is a logical version of the question:
is the TFNP class PLS different from FP? 
Here we at least have a relativized separation between PLS and FP~\cite{bk:boolean}, although
this implies nothing directly about the unrelativized theories.

\section{Traditional proof systems} \label{sec:trad}

We require that proof systems are sound 
and that refutations in a given system are recognizable in polynomial time.
When comparing two systems $P$ and $Q$ we are usually interested in their
behaviour when refuting CNFs, and we use the following basic definition.%

\begin{definition}\label{def:simulates}
We say that $Q$ \emph{simulates} $P$
if there is a polynomial-time function which, given a $P$-refutation of a CNF $\Gamma$,
outputs a $Q$-refutation of~$\Gamma$. 
$Q$ and $P$ are \emph{equivalent} if they simulate each other.
\end{definition}

Often it will make sense to discuss not only refutations of formulas,
but also \emph{derivations} of one formula from another.
We will use the notation e.g. ``a derivation $\Gamma \vdash \Delta$''
instead of ``a derivation of $\Delta$ from $\Gamma$'' and will write $\pi : \Gamma \vdash \Delta$
to express that $\pi$ is such a derivation.

\subsection{Quantified Boolean formulas and $G_1$} \label{sec:G1}
$G_1$ is a fragment of $G$, a proof 
system used for reasoning with quantified Boolean formulas. We give only a brief 
description of $G$ 
-- for more details see \cite[Chapter 4]{krajicek:proof-complexity}.
We will only consider fragments of $G$ as systems for refuting CNFs; for comparisons
of $G$ with some other systems in the context of proving quantified Boolean formulas, 
see e.g.~\cite{bp:understanding, ch:relating-powerful}. 

A \emph{quantified Boolean formula}, or QBF, is built from propositional variables and connectives
in the usual way, and also allows quantification over Boolean variables. That is,
if $F(x)$ is a QBF containing a 
variable $x$, then so are $\exists x \, F(x)$ and
$\forall x \, F(x)$.
In this context $\exists x$ and $\forall x$ are Boolean quantifiers
 and these formulas semantically have the expected meanings
$F(0) \vee F(1)$ and $F(0) \wedge F(1)$. 
We stratify QBFs into classes called $\Sigma^q_1$, $\Pi^q_1, \Sigma^q_2, \Pi^q_2$ etc. in the usual way, by counting
quantifier alternations. In particular, $\Sigma^q_1$ is the closure of the class of (quantifier-free) Boolean formulas
under $\vee, \wedge$ and $\exists$.
The strength of proof systems working with QBFs
is that they allow us to represent an ``exponential-size concept''
such as $\bigvee_{\vec a \in \{0,1\}^n} F(\vec a)$
with a polynomial-size
piece of formal notation $\exists x_1 \dots x_n F(\vec x)$.

The proof system $G$ is an extension of the propositional sequent calculus.
In this context a \emph{sequent} is an expression of the form 
\[
A_1, \dots, A_k \longrightarrow B_1, \dots, B_\ell
\]
where $A_1, \dots, A_k$ and $B_1, \dots, B_\ell$ are QBFs. Such a sequent is understood
semantically to mean the same as $\bigwedge_i A_i \rightarrow \bigvee_j B_j$,
and we say that an assignment \emph{satisfies} a sequent if it satisfies this formula. 
A derivation in $G$ is a sequence of sequents, 
each of which is either an axiom of the form $A \longrightarrow A$, or follows from one 
or two earlier sequents by one of the rules. These rules are sound and complete
and we will not list them, as the details are not so important for us
(see Section~\ref{sec:BA} for our justification of this).
For $k \in \NN$
the system $G_k$ is the restriction of $G$ which only allows formulas
from $\Sigma^q_k \cup \Pi^q_k$ to appear in derivations.

We are interested in proof systems as ways of refuting CNFs. To turn $G_1$
into such a system we have the following definition, where for
definiteness we think of $\Gamma$ as a single QBF (rather than, say,
as the cedent 
given by its clauses) and where $\bot$ 
is the empty cedent.

\begin{definition}
A \emph{$G_1$ refutation} of a CNF $\Gamma$ is a $G_1$ sequent calculus
derivation of the sequent $\Gamma \longrightarrow \bot$.
\end{definition}

\subsection{Pseudo-Boolean constraints and cutting planes} \label{sec:CP}

Following~\cite{bgmn:journal}, we use the term \emph{pseudo-Boolean constraint},
or \emph{PB constraint}, for a linear inequality with integer coefficients over $0/1$-valued
variables. We will sometimes call a set of such constraints a \emph{PB formula}.
PB constraints generalize clauses, since a clause 
$C = x_1 {\vee} \dots {\vee} x_n {\vee} \neg y_1 {\vee} \dots {\vee} \neg y_m$ 
can be expressed by a constraint of the form $x_1 + \dots + x_n + (1-y_1) + \dots + (1-y_m) \ge 1$.
We call this constraint $C^*$ and will also write $\Gamma^*$ for the
PB formula obtained by taking $C^*$ for each clause $C$ in a CNF~$\Gamma$.

If $C$ is a PB constraint $A\vec x \ge b$ we write $\neg C$ 
for the PB constraint $A \vec x \le b -1$.
Note that, although it is semantically the same, this denotes a different
piece of syntax from $\neg C$ when $C$ is a clause.

Given a substitution~$\omega$, we write $C_{\restrict \omega}$ 
for the PB constraint obtained by simply replacing each variable~$x$ in~$C$
with $\omega(x)$, and we use a similar notation for PB formulas.

We use \emph{cutting planes}
\cite{cook1987complexity}, or CP, as a derivational system for deriving one
PB formula~$G$ from another PB formula $F$.
A CP derivation is a sequence of PB constraints, including
every constraint from $G$, such that each constraint is either from $F$, or 
is a Boolean axiom of the form $x_i \ge 0$ or $x_i \le 1$, or
follows
from earlier constraints by one of the rules. 
These are \emph{addition} -- we can derive a new constraint
by summing integral multiples of two old constraints; and \emph{rounding}
-- from a constraint $d A \vec x \ge b$, where $d>0$ is an integral 
scalar and $A$ is an integral matrix of coefficients, we can derive $A \vec x \ge \lceil b/d \rceil$.
We use the notation $F \vCP G$ for CP derivations.
In a formal CP derivation, coefficients are written in binary.

\subsection{Extended resolution as a derivational system} \label{sec:ER_derivation}

For CNFs $\Gamma$ and $\Delta$, a \emph{resolution derivation}
$\Gamma \vdash \Delta$ is a sequence of clauses, beginning with the clauses
of $\Gamma$ and containing every clause from $\Delta$, such that each clause
in the sequence is either in the initial copy of $\Gamma$ or is derived from earlier clauses
by the \emph{resolution} or \emph{weakening} rule. 
Here the resolution rule derives $C \vee D$ from $C \vee x$ and $D \vee \neg x$,
for any variable~$x$, and the weakening rule derives $D$ from $C$ for any $D \supseteq C$
(although see below for a restriction on weakening in the context of extended resolution derivations). 
Because we allow propositional literals $0$ and $1$ to occur in clauses, 
we need to be able to remove them, so we also allow derivation of $C$ from $C \vee 0$
(this can be thought of as resolution with a notional axiom ``$1$'').
A \emph{resolution refutation} of $\Gamma$ is a derivation of $\Gamma \vdash \bot$.

An \emph{extension axiom} has the form of three clauses
$\neg u \vee \neg v \vee y$,
$\neg y \vee u$ and $\neg y \vee v$
which together express that $y$ is equivalent to $u \wedge v$. 
The intended use is that $u, v$ are literals and $y$ is a newly-introduced variable.
For good behaviour under restrictions we  also allow 
extension axioms of the form $\neg u \vee y$, $\neg y \vee u$ expressing that $y$
is equivalent to a single existing literal, and of the form $y$ or $\neg y$
expressing that $y$ is equivalent to a constant.

\begin{definition} \label{def:ER_derivation}
For CNFs $\Gamma$ and $\Delta$,
an \emph{extended resolution  ($\ER$)} derivation $\Gamma \vdash \Delta$ 
is a sequence of clauses, beginning with $\Gamma$ and including every clause in~$\Delta$.
Each clause in the sequence either appears in the initial copy of $\Gamma$, or is
derived from earlier clauses by resolution or weakening (where weakening is not
allowed to introduce a variable that has not appeared earlier in the sequence\footnote{
This restriction on  weakening is probably not strictly necessary.
We include it because it has the helpful consequence that each new variable comes with an explicit definition in terms of the old variables.}) or by the \emph{extension rule}, which allows us to introduce
an extension axiom defining a variable that has not appeared earlier in the sequence
from variables that have appeared earlier.
\end{definition}

Such a derivation is sound in the following sense: any assignment to all variables
in $\Gamma$ which satisfies $\Gamma$ can be extended to an assignment to
all variables in $\Delta$ which satisfies $\Delta$. The extension axioms in the derivation
tell us explicitly how to extend the assignment.

We will often need to handle many extension axioms at once:

\begin{definition} \label{def:extension_set}
Let $\vec x, \vec y$ be disjoint tuples of variables. We say that \emph{$\Delta$ 
is a set of extension axioms over $\vec x ; \vec y$} if it
can be written as a sequence of extension axioms defining variables $y_1, \dots, y_r$ in order,
where each $y_i$ is defined in terms of variables from among~$\vec x, y_1, \dots, y_{i-1}$.
\end{definition}

Equivalently, such a $\Delta$ can be thought of as describing a Boolean circuit and asserting
that, on input $\vec x$, the values computed at the internal nodes are $\vec y$.
We also introduce  notation for writing sets of extension axioms in this way:

\begin{definition} \label{def:extension_circuit}
Suppose $\vec x, \vec y$ are  tuples of  variables
 and $C$ is a circuit with gates of fan-in 2.
We write $[\vec y = C(\vec x)]$ for the set of extension
axioms over $\vec x ; \vec y$ expressing that the non-input nodes of the circuit
have values $\vec y$ on inputs $\vec x$ (we assume $\vec x, \vec y$ have suitable arities). 
If the circuit has a distinguished output
node we label the corresponding variable in $\vec y$ as $y^\out$.
\end{definition}

The condition that the extension rule must introduce new variables 
has some counter-intuitive consequences, and we must take extra care
 when we use $\ER$ as a derivational system.
For example, there are $\ER$ derivations $\top \vdash x$ 
and $\top \vdash \neg x$, (where $\top$ is the empty CNF)
but there is no derivation $\neg x \vdash x$, even though $\neg x$ extends~$\top$, 
and no derivation $\top \vdash x \wedge \neg x$. In many ways the new variables behave like existentially quantified bound variables.

Under some reasonable conditions on how extension variables
are used, we can avoid problems related to such issues. 
Below we formally prove two lemmas of this kind,
which we will refer to as needed.
We could avoid the issue by working with some other
system equivalent to $\ER$, such as extended Frege, circuit Frege~\cite{jerabek:dual-wphp} or even treelike~$G_1$,
but it would then be harder to show that the resulting system is simulated 
by linear dominance.

We use a convention that, 
in the context of $\ER$ derivations, when  we write an expression of the form
$\Gamma(\vec x) \vdash \Delta$ with some variables $\vec x$ displayed on the left,
we mean that every variable in $\vec x$ is treated as an ``old'' variable 
in this derivation and as such is not used as an extension variable in any instance of the extension rule,
even if it does not actually appear in the CNF~$\Gamma(\vec x)$.

\begin{lemma} \label{lem:compose_ER_derivations}
Let $\Gamma(\vec x), A(\vec x, \vec y), B(\vec x, \vec z), \Delta(\vec x, \vec w)$ by CNFs,
where we assume $\vec x$, $\vec y$, $\vec z$, $\vec w$ are disjoint and  no other
variables appear. Suppose we have ER derivations
\[
\pi_1 : \Gamma(\vec x) \wedge A(\vec x , \vec y) \vdash B(\vec x, \vec z)
\quad \textrm{and} \quad
\pi_2 : \Gamma(\vec x) \wedge B(\vec x, \vec z) \vdash \Delta(\vec x, \vec z, \vec w).
\]
Then we can construct  an ER  derivation
$\Gamma(\vec x) \wedge A(\vec x , \vec y) \vdash \Delta(\vec x, \vec z, \vec w)$ in polynomial time.
 \end{lemma}

\begin{proof}
We first copy $\pi_1$. Then we copy $\pi_2$,
except that every extension variable in $\pi_2$ which is not in $\vec w$ is given a new name,
to avoid clashes with variables $\vec y$ and other extension variables that appeared in $\pi_1$.
\end{proof}

\begin{lemma} \label{lem:ER_add_conclusions_to_assumptions}
Given an $\ER$ derivation $\pi_1 : \Gamma(\vec x) \vdash \Delta \wedge A$,
where $\Delta$ is a set of extension axioms over $\vec x; \vec y$,
we can construct in polynomial time an $\ER$ derivation 
$\pi_2 : \Gamma(\vec x) \wedge \Delta \vdash \Delta \wedge A$.
\end{lemma}

\begin{proof}
Let $\Delta'$ and $A'$ be the same as $\Delta$ and $A$ except that 
we have replaced every variable~$y_i$ with a new variable~$z_i$.
From $\Gamma \wedge \Delta$, we can derive $\Delta' \wedge A'$ by 
a copy of $\pi_1$ with $\Delta$ added to the initial clauses
and with each $y_i$ changed to $z_i$ everywhere outside of $\Delta$. 
Then we can work through $\vec y$ and derive, from the relevant 
extension axioms in $\Delta$ and $\Delta'$, 
using the normal rules of resolution, that $y_i \leftrightarrow z_i$;
formally, this is the two clauses $\neg y_i \vee z_i$ and $\neg z_i \vee y_i$.
Finally we resolve these clauses with the clauses of $A'$
to derive $A$.
\end{proof}

\section{The dominance rule} \label{sec:dominance}

We define three refutational proof systems using versions of the dominance-based
strengthening rule of~\cite{bgmn:journal}.
The \emph{dominance} proof system is
intended to be the same as the system described in~\cite{bgmn:journal}
 except that we have removed the machinery
for talking about optimization, that is, everything related to the objective function~$f$.
The \emph{linear dominance} proof system restricts this by only allowing 
a particular kind of ordering to be used in the dominance rule; 
the practical work in~\cite{bgmn:journal} in fact only needs this weaker
system.
Lastly we introduce our auxiliary system $\ERpls$.

In these  systems, often we can only apply a rule on the condition
that some other derivation $\Gamma \vdash \Delta$ or $\CC \vCP \DD$ exists, in 
$\ER$ or $\CP$, 
possibly involving formulas that 
do not appear explicitly in the proof
we are working on;
or that some other  polynomial-time-checkable 
object exists, such as 
a substitution $\omega$.
To ensure that correctness of a proof is checkable in polynomial time 
 we implicitly require that, in a formal proof in a dominance-based system, each application of the rule
is labelled with an example of the object in question, 
with the size of the labels (that is, the CP derivations, substitutions etc.)
counted towards the size of the formal proof.

\subsection{The dominance proof system} \label{sec:dominance_system}

This is a system for refuting PB formulas. 
As in~\cite{bgmn:journal}, we will call steps in a refutation \emph{configurations},
rather than lines. A configuration is a quadruple 
$( \CC, \DD, \OOl, \vec z )$ where
\begin{itemize}
\item
$\CC$ is a set of PB constraints called \emph{core constraints}
\item
$\DD$ is a set of PB constraints called \emph{derived constraints}
\item
$\OOl(\vec x, \vec y)$ is a PB formula 
where $\vec x$ and $\vec y$ both have the same arity as $\vec z$
\item
$\vec z$ is a tuple of variables.
\end{itemize}
We put no conditions on which variables appear in $\CC$ and $\DD$,
except that the variables $\vec x, \vec y$ in $\OOl(\vec x, \vec y)$ should be thought
of as dummy variables that are not related to the rest of the proof.
In a valid proof, in every configuration the formula $\OOl(\vec x, \vec y)$
defines a preorder and we use this with $\vec z$ to define a preorder~$\preceq$ on assignments,
writing $\alpha \preceq \beta$ if $\OOl(\vec x, \vec y)$ is satisfied under the assignment
that takes $\vec x$ to $\alpha(\vec z)$ and $\vec y$ to $\beta(\vec z)$.

Semantically a configuration can be thought of as asserting that $\CC$ is satisfiable,
and that if we order assignments by $\OOl$ on $\vec z$ as described above, then for any
assignment $\alpha$ satisfying~$\CC$, some assignment $\beta$
with $\beta \preceq \alpha$ satisfies $\CC \cup \DD$
(see Definition~\ref{def:validity} below).

A refutation of a PB formula $F$ is then a sequence of configurations,
beginning with
$(F, \emptyset, \top, \emptyset)$,
where $\top$ is the empty PB formula,
and ending with a configuration in which 
$\CC$ or $\DD$ contains the contradiction $\bot$,
that is, $0 \ge 1$. Each configuration is derived from 
the previous configuration $(\CC, \DD, \OOl, \vec z)$
by one of the following rules:

\paragraph{Implicational derivation rule.}
Derive 
$(\CC, \DD \cup \{ C \}, \OOl, \vec z)$,
if there is a derivation\footnote{In \cite{bgmn:journal}, the derivation is allowed to use some additional inferences beyond those of $\CP$.
For simplicity we omit these,
as in the presence of the redundance-based strengthening rule, even strengthening
$\CP$ here to a system like extended Frege would not make any
difference to the overall dominance system. In particular, our proof of the simulation of
 linear dominance  by $G_1$ in Section~\ref{sec:soundness} would still go through.}
 \mbox{$\CC \cup \DD \vCP C$}.

\paragraph{\textrm{(}Objective bound update rule}\!\!\!\!-- this appears in~\cite{bgmn:journal},
but we omit it from our systems as it only affects the objective function~$f$, which we do not use.)

\paragraph{Redundance-based strengthening rule.}
Derive $(\CC, \DD \cup \{ C \}, \OOl, \vec z)$
if there is a substitution~$\omega$ and a derivation
$
\CC \cup \DD \cup \{ \neg C \}
\vCP
(\CC \cup \DD \cup \{ C \})_{\restrict \omega} \cup \OOl (\vec z_{\restrict \omega}, \vec z).
$

\paragraph{Deletion rule.}
Derive  $(\CC', \DD', \OOl, \vec z)$ if
\begin{enumerate}
\item
$\DD' \subseteq \DD$ and
\item
$\CC'=\CC$ or $\CC' = \CC \setminus \{ C \}$ for some constraint $C$
derivable by the redundance rule above from $(\CC', \emptyset, \OOl, \vec z)$\footnote{This restriction of the deletion rule ensures that it preserves semantic validity under the intuitive meaning of configurations mentioned above. See Section \ref{sec:soundness} for an argument.}.
\end{enumerate}

\paragraph{Transfer rule.}
Derive $(\CC', \DD, \OOl, \vec z)$ if
$\CC \subseteq \CC' \subseteq \CC \cup \DD$.
In other words, we can copy constraints from $\DD$ to $\CC$

\paragraph{Dominance-based strengthening rule.}
We first give a slightly informal definition:
 derive $(\CC, \DD \cup \{ C \}, \OOl, \vec z)$
if there is a substitution $\omega$ and, informally, a derivation
\[
\CC \cup \DD \cup \{ \neg C \}
\vCP
\CC_{\restrict \omega} \cup (\vec z_{\restrict \omega} \prec \vec z)
\]
where $\vec z_{\restrict \omega} \prec \vec z$ expresses that $\vec z_{\restrict \omega}$
is strictly smaller than $\vec z$ in the ordering $\OOl$.
However it may be  that any PB formula expressing the strict inequality $(\vec z_{\restrict \omega} \prec \vec z)$
is very large.
So formally the rule is: derive 
$(\CC, \DD \cup \{ C \}, \OOl, \vec z)$ if there is a substitution 
$\omega$ and two derivations
\begin{align*}
 \CC \cup \DD \cup \{ \neg C \}
& \vCP
\CC_{\restrict \omega} \cup \OOl(\vec z_{\restrict \omega}, \vec z) \\
\CC \cup \DD \cup \{ \neg C \} \cup \OOl(\vec z, \vec z_{\restrict \omega}) & \vCP \bot.
\end{align*}

\paragraph{Order change rule.}
From $(\CC, \emptyset, \OOl, \vec z)$ derive 
$(\CC, \emptyset, \OOl', \vec z')$ if $\OOl'$ is CP-provably a preorder.
That is, if there are derivations
$\emptyset  \vCP \OOl(\vec u, \vec u)$
and $\OOl(\vec u, \vec v) \cup \OOl(\vec v, \vec w)  \vCP \OOl (\vec u, \vec w)$.

\subsection{The linear dominance proof system} \label{sec:linear_dominance_system}

This  restricts the dominance proof system 
to only use orderings $\OOl$ arising from a multilinear objective function.
Formally, we require that $\OOl(\vec x, \vec y)$ is
always of the form $f(\vec x) \le f(\vec y)$, where $f$ is a multilinear function
$\vec x \mapsto \sum_i b_i x_i$ for some constants~$b_i$.
These constants can be changed using
the order change rule.
We are no longer required to explicitly 
prove that $\OOl(\vec x, \vec y)$ is an ordering,
as CP can always prove this for this restricted form.
 
The most important ordering of this form is the lexicographic ordering, 
which we get by setting $f(\vec x) \mapsto \sum_i 2^i x_i$
(for a suitable ordering of the variables in $\vec x$).

\subsection{The system $\ERpls$} \label{sec:ERpls}

This system uses the clausal version of the dominance rule
sketched in Proposition~\ref{pro:informal_rule} in the introduction.
The name is intended to suggest  that it has a similar connection to polynomial local search ``computations''
as $\ER$ has to polynomial time. 

We fix a polynomial-time constructible family of CNFs defining lexicographic ordering.
That is, for each $k$ we have a CNF $[x_1, \dots, x_k \lelex y_1, \dots, y_k]$,
which may also use some auxiliary variables $\vec z$, such that for all
$\alpha, \beta \in \{0,1\}^k$ there is an assignment to $\vec z$
satisfying $[\alpha \lelex \beta]$
 if and only if $\alpha \le \beta$ lexicographically.
It is not too important which CNF we use for $[\vec x \lelex \vec y]$,
but we specify a convenient one in Appendix~\ref{sec:ordering} below.

An $\ERpls$ refutation of a CNF $\Gamma$ is formally a sequence of CNFs, beginning with
$\Gamma$ and ending with a CNF containing the empty clause~$\bot$.
At each step we apply one of the two rules below to derive the next CNF in the sequence.

\paragraph{ER rule.} From $\Gamma$ derive $\Gamma \wedge \Delta$ 
if there is an $\ER$ derivation  $\Gamma \vdash \Delta$, 
in the sense of Definition~\ref{def:ER_derivation}.

\paragraph{Dominance rule.} 
Let $\vec x$ list all variables of $\Gamma$ in some order and let $C$ be a clause in these variables.
From $\Gamma$ derive $\Gamma \wedge C $,
provided we have
\begin{enumerate}
\item
a set $\Delta$ of extension axioms over $\vec x ; \vec y$
\item
a substitution $\omega$ mapping variables  $\vec x$ to variables among $\vec x \cup \vec y$
\item\label{it:erpls-cond-3}
two $\ER$ derivations
\begin{enumerate}
\item
$\Gamma \wedge \Delta \wedge  \neg C \vdash \Gamma_{\restrict \omega}$
\item
$\Gamma \wedge \Delta \wedge  \neg C \wedge [\vec x \lelex \vec x_{\restrict \omega}] \vdash \bot$
\end{enumerate}
\end{enumerate}
with the technical condition that
the auxiliary variables $\vec z$ used in $[\vec x \lelex \vec x_{\restrict \omega}] $
may not appear in $\Gamma$, $\Delta$, or $C$.

\medskip

Informally,  condition~3 can be thought of as asking for a single ER derivation
$
\Gamma \wedge \Delta \wedge \neg C  
\vdash
\Gamma_{\restrict \omega} \, \wedge \, [ \vec x_{\restrict \omega} \! <_\mathrm{lex} \! \vec x ],
$
as in Proposition~\ref{pro:informal_rule}.
Note that we do not have any deletion rule, that is, 
we can only grow working set of clauses $\Gamma$, and never shrink it. 
This is  because $\Gamma$ is modelled on the \emph{core}
constraints $\CC$ in the dominance system, which can only be deleted 
in very specific 
situations. For simplicity we do not allow deletion at all, 
as we will not need it.

\begin{lemma} \label{lem:ERpls_soundness}
If $\Gamma'$ is derived from $\Gamma$ by a rule of $\ERpls$,
then $\Gamma'$ and $\Gamma$ are equisatisfiable.
\end{lemma}

\begin{proof}
The only nontrivial case is the forward direction of the dominance rule. 
This is proved in the same way as Proposition~\ref{pro:informal_rule}
in the introduction, with the cosmetic change that we now have
$[ \vec x \lelex \vec x_{\restrict \omega} ]$
on the left of the entailment rather than 
$[ \vec x_{\restrict \omega} \! <_\mathrm{lex} \! \vec x ]$ on the right. We must 
also deal now with the auxiliary variables in  $[ \vec x \lelex \vec x_{\restrict \omega} ]$,
but since these are not in the domain of the ordering this presents no problem.
\end{proof}

\section{Bounded arithmetic} \label{sec:BA}

We will carry out some arguments in  theories of bounded arithmetic,
which we will turn into propositional proofs using variants of well-known 
translations. Here we give a brief overview 
 -- for more see e.g. \cite[Chapters 9 and 12]{krajicek:proof-complexity}.
When we write that a formula with free variables is provable in a first-order theory, we mean that its universal
closure is provable.

\subsection{Theories}\label{subsec:ba}

$\PV$ is the canonical theory for polynomial-time reasoning~\cite{cook:feasibly}. 
Its language contains a function symbol, called a \emph{$\PV$ function}, for every polynomial-time algorithm
on~$\NN$.
Its axioms are defining equations for all $\PV$ functions, based
on Cobham's characterization of polynomial-time functions.
See e.g. \cite[Chapter 12.1]{krajicek:proof-complexity} for a precise definition
(there the theory is called~$\PV_1$). 
Importantly, $\PV$ proves the principle of mathematical induction 
for any property defined by a \emph{$\PV$ formula} -- that is, by a quantifier-free
formula in the language of $\PV$. 
Such formulas define precisely the polynomial-time properties.

More powerful theories  can be obtained by extending $\PV$ with stronger induction axioms. 
A formula in the language of $\PV$ is $\Sigma^b_1$ 
if it has the form $\exists x \! \le \! t\, \varphi$, 
where $t$ is a term not containing $x$ and $\varphi$ is a $\PV$ formula; unsurprisingly,
$\Sigma^b_1$ formulas define exactly properties in NP. A formula is $\Sigma^b_2$
if it has the form $\exists x \! \le \! t_1\,\forall y \! \le \! t_2\, \varphi$ for $\varphi$ a $\PV$ formula.
The classes  ${\Pi}^b_1$ and ${\Pi}^b_2$ are defined dually.

The theory $\mathit{T}^1_2$ (a more accurate name would
be ${T^1_2}(\PV)$) 
extends $\PV$ by induction axioms for all $\Sigma^b_1$ formulas.
$T^1_2$ is the weakest theory that suffices to prove the least number principle for $\Sigma^b_1$ formulas, 
that is,  that every nonempty NP set has a least element; actually, even
the least number principle for $\PV$ formulas already implies $T^1_2$ over $\PV$ \cite{buss:thesis}.

The theory $\mathit{S}^1_2$, intermediate in strength between $\PV$ and $T^1_2$,
extends $\PV$ by the \emph{length induction} axioms for $\Sigma^b_1$ formulas, 
that is, universal closures of statements of the form
\[
\psi(0) \wedge \forall x\, (\psi(x) \to \psi(x+1)) \to \forall x\,\psi(|x|)), \ 
\]
where $\psi$ is $\Sigma^b_1$ and $|\cdot|$ stands for the \emph{length} function
that takes a number $x$ to its length in binary notation. 
The theory $\mathit{S}^2_2$ is a strengthening of $T^1_2$ that additionally contains
length induction for $\Sigma^b_2$ formulas. It should be noted that
$S^1_2$ proves length induction also for $\Pi^b_1$ formulas,
$T^1_2$ proves induction also for $\Pi^b_1$ formulas, and so on.

Ordered by strength, we have 
$\PV \subseteq S^1_2 \subseteq T^1_2 \subseteq S^2_2 \subseteq \dots$.
There is also partial conservativity between some adjacent theories.
In particular, $S^1_2$ is {$\forall \Sigma^b_1$-conservative} over $\PV$ and
$S^2_2$ is {$\forall \Sigma^b_2$-conservative} over $T^1_2$~\cite{buss:axiomatizations}.
This means that if $\psi(x)$ is a $\Sigma^b_1$ formula and $S^1_2$ proves
$\forall x\, \psi(x)$, then $\PV$ proves it as well; analogously for
$\Sigma^b_2$ formulas, $S^2_2$ and $T^1_2$.

\medskip

There is a well-known connection between  propositional proof systems and
arithmetic theories,
linking $\ER$ to $\PV$ (and $S^1_2$) and $G_1$ to $T^1_2$
(and $S^2_2$). 
The following theorem, which shows that
a theory proves the soundness of the corresponding proof system,
can be viewed as an upper bound: it says that, for example, 
$G_1$ is in some sense no stronger than~$T^1_2$.
We use this for our main result in Section~\ref{sec:ERpls_simulates_G1}.

\begin{theorem}[\cite{cook:feasibly, kp:quantified-propositional}] \label{the:provable_G1_soundness}
$S^1_2$ -- and by conservativity, $\PV$ -- proves the CNF-reflection principle for $\ER$: 
``any CNF refutable in $\ER$ is unsatisfiable''. 
Similarly, $S^2_2$ -- and by conservativity, $T^1_2$ -- proves the CNF-reflection principle for
$G_1$: ``any CNF refutable in $G_1$ is unsatisfiable''.
\end{theorem}

\begin{proof}[Proof sketch]
We first consider~$G_1$. 
The statement ``sequent $s$, containing only formulas from $\Sigma^q_1 \cup \Pi^q_1$,
is satisfied by assignment $\alpha$ to its free variables'' can be naturally
written as a $\Pi^b_2$ formula~$\sigma(s, \alpha)$.
Using this we formalize the natural proof of the soundness of $G_1$ 
as a length induction on a $\Pi^b_2$ formula,
roughly as follows: given a $G_1$ derivation~$\pi$,
we show by induction down~$\pi$ that 
$\forall \alpha\, \sigma(s, \alpha)$ holds for every line~$s$ of~$\pi$. 
If the last line has the form $\Gamma \vdash \bot$ for a CNF $\Gamma$, 
this means that $\Gamma$ cannot be satisfiable.

The part about $\ER$ is proved by a similar argument, but with the length induction hypothesis being
``every line in the $\ER$ refutation $\pi$ up to the current one is satisfiable''. This can be stated
in a $\Sigma^b_1$ way, holds at the beginning of $\pi$ if the CNF $\Gamma$ being refuted is satisfiable,
but no longer holds once $\pi$ reaches the empty clause $\bot$. 
\end{proof}

In the other direction, it is possible to translate proofs in an arithmetic theory
into uniform families of propositional proofs. 
We will use this translation  for $\PV$ and $\ER$,
and we give a slight refinement of it in Section \ref{subsec:prop-transl}.
Such  translations can also be used to give a kind of converse
to Theorem~\ref{the:provable_G1_soundness}, 
with a proof similar to our approach in Section~\ref{sec:ERpls_simulates_G1} below:

\begin{theorem} [\cite{cook:feasibly, kp:quantified-propositional}] \label{thm:t12-g1}
 If $\PV$ (equivalently~$S^1_2$) proves
the CNF-reflection principle for a propositional proof system $Q$, then $\ER$ simulates $Q$.
Similarly, if $T^1_2$ (equivalently~$S^2_2$) proves the CNF-reflection principle for $Q$, 
then $G_1$ simulates $Q$.
\end{theorem}

Theorems~\ref{the:provable_G1_soundness} and~\ref{thm:t12-g1} together give us a close
association between $T^1_2$  and short $G_1$ proofs, and in fact we prove our results about $G_1$ 
indirectly using these theorems, rather than by reasoning about~$G_1$ itself. 
This is largely the reason that we did not include a complete description of $G_1$ in Section~\ref{sec:G1}.

\medskip

Another important property of $T^1_2$ -- and, by conservativity, of $S^2_2$ -- is that its $\forall \Sigma^b_1$ consequences
can be witnessed by \emph{polynomial local search}. We formally define a PLS problem as 
a triple $(t_w, \theta_w, N_w)$, where $t, N$ are respectively a unary and a binary $\PV$ function, $\theta$~is a binary $\PV$ formula, and the distinguished argument $w$ written in the subscript is an instance of the problem. The formula $\theta$ defines the domain of the problem on instance $w$ (tacitly, any element of the domain is required to be at most polynomially larger than~$w$); $t_w$ is an initial value that should be in the domain; and $N_w$ is a one-place \emph{neighbourhood} function, 
which attempts to map any 
value in the domain to a strictly smaller value in the domain. Since the domain has a least element as long as it is nonempty, $N_w$ will sometimes fail, and a solution to the problem on instance $w$ is either $t_w$, if $\neg \theta_w(t_w)$, or a value $y$ such that $\theta_w(y)$ but either
$\neg \theta_w(N_w(y))$ or $N_w(y) \ge y$. 

The PLS witnessing theorem for $T^1_2$, originally proved in \cite{bk:boolean}, 
says that if $T^1_2$ proves $\forall w \, \exists y\! \le \! t \, \varphi$, 
where $\varphi$ is a $\PV$ formula, 
then the task of finding~$y$ given~$w$ can be reduced to a PLS problem.
Written in a modern form, which also includes an upper bound on the strength of the theory
 needed to prove correctness of the reduction, we have:
 
\begin{theorem}[{\cite[Theorem 2.5]{bb:pls-ph}}] \label{thm:pls}
Assume that $T^1_2 \vdash \forall w \, \exists y\! \le \! t \, \varphi(w,y)$, where $\varphi$ is a $\PV$ formula. Then
there is a PLS problem $Q_w =(t_w, \theta_w, N_w)$ and a $\PV$ function $f$ such that the following are provable in $\PV$:
\begin{enumerate}
\item
 $\neg \theta_w(t_w)  \rightarrow  \phi(w,f(w)) $
 \item
 $\theta_w(z) \wedge \neg \theta_w(N_w(z)) \ \rightarrow \phi(w,f(z))$
\item
 $\theta_w(z) \wedge N_w (z) \ge z  \rightarrow  \phi(w,f(z))$.
\end{enumerate}

\end{theorem}

\subsection{Propositional translations}\label{subsec:prop-transl}

We use a version of the translation from $\PV$ proofs to 
polynomial-time constructible families of $\ER$ proofs, due to Cook~\cite{cook:feasibly}.%
\footnote{We emphasize that we are using the Cook translation,
rather than the Paris-Wilkie translation of e.g.~\cite{pw:counting}.
The Paris-Wilkie translation is usually used to translate
first-order proofs involving an oracle symbol into families of small proofs in relatively weak
propositional systems. For example, it translates (a relativized version of) $T^1_2$ into polylogarithmic-width resolution.}
We first describe how to translate formulas.
Consider a $\PV$ formula $\theta(\vec x)$. Let $\vec k$ represent
 a choice of binary bit-lengths for the variables $\vec x$
 (we will not be very formal about~$\vec k$; we can think of this notation as assigning 
 a length to every free first-order variable in the universe). 
Supposing $\vec x = x_1, \dots, x_\ell$, we will code each $x_i$ using
a $k_i$-tuple of new propositional variables $x^1_i, \dots, x^{k_i}_i$, which tuple
we write just as $\vec x_{i}$.

By definition, $\theta$ is a quantifier-free formula built from $\PV$ functions.
So we can construct, in some canonical way based on the structure of~$\theta$,
in time polynomial in the bit-lengths $k_1, \dots, k_\ell$, 
a Boolean circuit~$C(\vec x_1, \dots, \vec x_\ell)$ 
evaluating $\theta$ on binary inputs of these lengths.
Following Definition~\ref{def:extension_circuit}, we introduce
a tuple of new variables $\vec z$, one for each node in~$C$,
and define the propositional translation
\mbox{$\norm{\theta(\vec x)}_{\vec k}$} to be the CNF
$[\vec z = C(\vec x_1 , \dots, \vec x_\ell)] \wedge z^\out$.
Unless stated otherwise, we assume that the translations of any two explicitly listed
formulas have disjoint auxiliary variables~$\vec z$.
For example, in Proposition~\ref{pro:PV_to_ER},
we assume that the translations of $\phi_1, \dots, \phi_r, \theta$
all have disjoint auxiliary variables, even if some formula appears twice in this list.

We now state how we translate proofs. For the proof see Appendix~\ref{appendix_translations}.

\begin{prop} \label{pro:PV_to_ER}
Suppose $\PV$ proves 
a sentence 
$\forall \vec x, \, \phi_1(\vec x) \wedge \dots \wedge \phi_r(\vec x)
\rightarrow \theta(\vec x),$
where $\phi_1, \dots, \phi_r, \theta$ are quantifier-free. 
Then for any assignment $\vec k$ of bit-lengths to the variables $\vec x$,
we can construct in time polynomial in $\vec k$
an $\ER$ derivation 
\[
\norm{\phi_1(\vec x)}_{\vec k} \wedge \dots \wedge \norm{\phi_r(\vec x)}_{\vec k}
\vdash \norm{ \theta(\vec x)}_{\vec k}.
\]
\end{prop}


\section{$\ERpls$ simulates $G_1$} \label{sec:ERpls_simulates_G1}

This section contains the main technical work of the paper.
We begin by constructing some $\PV$ proofs. We will use propositional
translations of these in our simulation.

Let $\Sat(a,x)$ be a natural $\PV$ formula expressing that CNF $a$ 
is satisfied by assignment~$x$. Let $\Refu(a,b)$ be a natural $\PV$ formula
expressing that $b$ is a $G_1$ refutation of $a$.
We may take $\forall a,b,x \ \neg \Sat(a,x) \vee \neg \Refu(a,b)$ as the
CNF-reflection principle for
$G_1$, stating that any CNF refutable in $G_1$ is unsatisfiable. 
By Theorem~\ref{the:provable_G1_soundness} this is provable in $T^1_2$.

CNF-reflection is a universally-quantified $\PV$ formula, so in particular it is $\forall \Sigma^b_1$.
Thus by Theorem~\ref{thm:pls}, 
there is a PLS problem $Q_w =(t_w, \theta_w, N_w)$, where to
save space we think of $a,b,x$ as combined into a single parameter 
$w$ which we write as a subscript, such that the existence of a solution to~$Q_w$
 implies~$\neg \Sat(a,x) \vee \neg \Refu(a,b)$, provably in $\PV$.
Precisely, $\PV$ proves the following three formulas, in free variables~$a,b,x,y$
(note that since the CNF-reflection principle does not contain an existential quantifier, we do not need 
the function $f$ that appears in Theorem~\ref{thm:pls}):
\begin{enumerate}
\item
$\neg \theta_w(t_w) \ \rightarrow \  (\neg \Sat(a,x) \vee \neg \Refu(a,b))$
\item
$\theta_w(y) \wedge \neg \theta_w(N_w(y)) \ \rightarrow \  (\neg \Sat(a,x) \vee \neg \Refu(a,b))$
\item
$\theta_w(y) \wedge N_w (y) \ge y \ \rightarrow (\neg \Sat(a,x) \vee \neg \Refu(a,b))$.
\end{enumerate}
By standard properties of PLS, 
we may assume in order to simplify some things below that the bit-length of $t_w$ depends only on the components $a,b$ of $w$ and not on the assignment~$x$, and that $N_w$ is hard-wired to never give output bigger than~$t_w$.

Making some small rearrangements and introducing a new variable $u$ for
the neighbour of $y$,
we get that $\PV$ proves
\begin{align*}
F1. \quad & \Sat(a,x) \wedge \Refu(a,b) \wedge y =t_w  \ \rightarrow \ \theta_w(y) \\
F2. \quad  & \Sat(a,x) \wedge \Refu(a,b) \wedge \theta_w(y) \wedge u = N_w(y) \ \rightarrow \ \theta_w(u) \\
F3. \quad  & \Sat(a,x) \wedge \Refu(a,b) \wedge \theta_w(y) \wedge u = N_w(y) \wedge y \le u \ \rightarrow \ \bot.
\end{align*}
With these proofs in hand we can  describe the simulation.

\begin{theorem} \label{the:ERpls_simulates_G1}
$\ERpls$ simulates $G_1$.
\end{theorem}

\begin{proof}
We are given a CNF $A$ and a $G_1$~refutation $B$ of $A$. We
want to construct, in polynomial time, an $\ERpls$ refutation of $A$.
We will build the refutation using propositional translations of the proofs $F1$-$F3$
above. We begin by calculating the bit-length of the variables $a,b,x,y,u$ which we will use in the translation.

Let $n$ be the number of variables in $A$ and let $m$ and $\ell$
be the bit-length of the strings coding $A$ and $B$ respectively.
We may assume $n \le m$.
We will use $m, \ell, n$ as the respective bounds on the bit lengths
of $a,b,x$.
By our simplifying assumption on the problem $Q_w$
we can find a bound $r$, polynomial in $m$ and $\ell$, on the bit-length of $t_w$ 
for parameters $w$ of the lengths we are considering.
We may also use~$r$ as the bit-length for both $y$ and $u$,
since $N_w$ needs at most $r$ bits to encode its output.
Thus we use these bounds $m, \ell, n, r, r$ as the bit-length parameter $\vec k$ in 
all our propositional translations below. 
As a result, all the CNFs we obtain from the translation 
will have size polynomial in $m + \ell$.
For simplicity of notation,
we will omit actually writing the subscript $\vec k$.

Applying Proposition~\ref{pro:PV_to_ER} to $F1$, $F2$ and $F3$,
we obtain, in time polynomial in~$m + \ell$, the following 
$\ER$ derivations:
\begin{align*}
P1: \quad & \norm{\Sat(a,x)} \wedge \norm{\Refu(a,b)} \wedge \norm{y =t_w}  
	\ \vdash \ \norm{\theta_w(y)} \\
P2: \quad  & \norm{\Sat(a,x)} \wedge \norm{\Refu(a,b)} \wedge \norm{\theta_w(y)} \wedge \norm{u = N_w(y)}
	\ \vdash \ \norm{\theta_w(u)} \\
P3: \quad  & \norm{\Sat(a,x)} \wedge \norm{\Refu(a,b)} \wedge \norm{\theta_w(y)} \wedge \norm{u = N_w(y)} \wedge \norm{y \le u} 
	\ \vdash \ \bot.
\end{align*}
These formulas and derivations are in propositional variables
$\vec a, \vec b, \vec x, \vec y, \vec u$  that arise from $a,b,x,y,u$ in the translation, using the bit-lengths described above (plus the requisite auxiliary and extension
variables).

Now let $\tau$ be the substitution (in fact a partial assignment) that
first replaces the propositional variables $\vec a$ and $\vec b$
with the actual bits of $A$ and $B$. 
Then, recalling that the CNFs $\norm{\Sat(a,x)}$ and $\norm{\Refu(a,b)}$
are defined in terms of Boolean circuits taking $\vec a,\vec b, \vec x$
as input, we compute the values of all nodes in these circuits that do not 
depend on $\vec x$ (for $\norm{\Refu(a,b)}$
this means all nodes) and let $\tau$ assign those values to the corresponding auxiliary variables in the CNFs.

Applying $\tau$ to the derivations above,
each of $P1_{\restrict \tau}, P2_{\restrict \tau}, P3_{\restrict \tau}$ is still a valid ER refutation.
However we may delete $\norm{\Refu(a,b)}_{\restrict \tau}$ from the assumptions,
since by construction $\tau$ satisfies every clause in $\norm{\Refu(a,b)}$,
because $B$ is in fact a refutation of $A$. 
Furthermore we make the following claims, where all the circuits and 
derivations asserted to exist are constructible in polynomial time from $A$ and $B$.
For the notation $[\vec z = C(\vec e)]$ see Definition~\ref{def:extension_circuit}.

\begin{enumerate}
\item
Write $A(\vec x)$ for the CNF $A$ with the variables
renamed to $x_1, \dots, x_n$. 
Then there is an ER derivation $A(\vec x) \vdash \norm{\Sat(a,x)}_{\restrict \tau}$.
\item 
There is a circuit $D_\theta$ and auxiliary variables $\vec z_y$ and $\vec z_u$ such
that 
\begin{enumerate}
\item
$\norm{\theta_w(y)}_{\restrict \tau}$
has the form
$[\vec z_y = D_\theta (\vec x, \vec y)] \wedge z_y^\out$
\item
$\norm{\theta_w(u)}_{\restrict \tau}$
has the form
$[\vec z_u = D_\theta (\vec x, \vec u)] \wedge z_u^\out$.
\end{enumerate}
Abusing notation, we may write these as $[\hat \theta(\vec x, \vec y)]$ and $[\hat \theta(\vec x, \vec u)]$.
\item
There is a circuit $\hat t$ and auxiliary variables $\vec z_t$ such that
there is an ER derivation $[(\vec y, \vec z_t) = \hat t (\vec x)] \vdash \norm{y =t_w}_{\restrict \tau}$.
\item
There is a circuit $\hat N$ and auxiliary variables $\vec z_N$ such that
there is an ER derivation $[(\vec u, \vec z_N) = \hat N (\vec x, \vec y)] \vdash \norm{u =N_w(y)}_{\restrict \tau}$.
\item
There is an ER derivation $[ \vec y \lelex \vec u ] \vdash \norm{ y \le u }_{\restrict \tau}$.
\end{enumerate}

Here claim~2 is true by construction,
and for claim~5 note that the CNF $\norm{ y \le u }$ is not changed 
after restricting by~$\tau$. Otherwise we appeal to the well-known strength and
robustness of $\ER$ 
and the fact that we are able to choose how
to express $\Sat$, and even $y \le x$, in $\PV$. 
We give more details of claim 1
in Appendix~\ref{sec:formalize_sat}.

Combining the derivations provided by the claims 
with $P1_{\restrict \tau}, P2_{\restrict \tau}, P3_{\restrict \tau}$
and appealing to Lemma~\ref{lem:compose_ER_derivations}
we get that the following three $\ER$ derivations can be constructed in 
polynomial time from $A$ and $B$.
\begin{align*}
R1: \quad & A(\vec x) \wedge [(\vec y, \vec z_t) = \hat t (\vec x)] \ \vdash \
	[\hat \theta(\vec x, \vec y)] \\
R2: \quad & A(\vec x) \wedge [\hat \theta(\vec x, \vec y)]
	\wedge [(\vec u, \vec z_N) = \hat N (\vec x, \vec y)] \ \vdash \
	[\hat \theta(\vec x, \vec u)] \\
R3: \quad & A(\vec x) \wedge [\hat \theta(\vec x, \vec y)]
	\wedge [(\vec u, \vec z_N) = \hat N (\vec x, \vec y)] 
	\wedge [ \vec y \lelex \vec u] \ \vdash \ \bot.
\end{align*}

For example, 
for $R2$ we first combine the derivation in claim~1, the first identity in claim~2, 
and the derivation in claim~4 to derive the LHS of $P2_{\restrict \tau}$
from the LHS of $R2$ (we may need to rename some variables
introduced by the extension rule to avoid clashes when we combine derivations,
as in the proof of Lemma~\ref{lem:compose_ER_derivations}). 
By Lemma~\ref{lem:compose_ER_derivations} we can then use
$P2_{\restrict \tau}$ to derive $\norm{\theta_w(u)}_{\restrict \tau}$, which
is precisely the RHS of $R2$ by the second identity in claim 2.

The reader should note that
R1--R3 still capture the same idea that we
began this section with, but now in a nonuniform version: 
they constitute a proof in ER that
if $\vec y$ is a solution
of a PLS problem related to $Q_w$, 
where $w$ is the instance $(A, B, \vec x)$, 
then $A(\vec x)$ is false.

We can now describe an $\ERpls$ refutation of $A(\vec x)$, 
and thus one of $A$. It will use one application of the ER rule
and one of the dominance rule.
We begin with $A(\vec x)$. We then introduce the clauses
$[\hat \theta(\vec x, \vec y)]$ by the ER rule. This is allowed, because
we can obtain them from $A(\vec x)$ by the following ER derivation: first write down the 
extension axioms~$[(\vec y, \vec z_t) = \hat t (\vec x)]$, 
then use $R1$.

To finish the refutation we  use the dominance rule to derive the empty clause.
That is, in the rule we take the new clause $C$ to be empty. 
The other ingredients are as follows.
\begin{enumerate}
\item
We set $\Gamma := A(\vec x) \wedge [\hat \theta(\vec x, \vec y)]$, so
$\Gamma$ consists of all the clauses we have so far.
\item
We let $\vec v := \vec y, \vec z_y, \vec x$ list all variables  that occur in $\Gamma$.
Here we deliberately put~$\vec y$ first so that it is most 
significant in determining the lexicographic order of assignments to $\vec v$.
\item
We set $\Delta := 
[(\vec u, \vec z_N) = \hat N (\vec x, \vec y)] \wedge [\vec z_u = D_\theta (\vec x, \vec u)]$.
This is a set of extension axioms over $\vec v; \vec u, \vec z_u, \vec z_N$.
\item 
We set $\omega$ to be the substitution which maps each variable in $\vec y, \vec z_y$
to the corresponding variable in $\vec u, \vec z_u$ and is the identity everywhere else.
\end{enumerate}
The substitution $\omega$ is chosen so 
that $\Gamma_{\restrict \omega} = A(\vec x) \wedge [\hat \theta(\vec x, \vec u)]$.
Also $\Delta$ and $\omega$ are chosen so that
the range of $\omega$ is a subset of the variables appearing in~$\Delta$,
as required by the rule.
 
The reader should have in mind the following informal
process, as sketched in the introduction. 
Suppose we have an assignment $\alpha$ to $\vec v$ satisfying 
$\Gamma = A(\vec x) \wedge [\hat \theta(\vec x, \vec y)]$. 
By R2, we can use the circuits described in $\Delta$ to extend it to an assignment $\alpha \cup \beta$
to $\vec v, \vec u, \vec z_u, \vec z_N$ satisfying $\Gamma_{\restrict \omega} = A(\vec x) \wedge [\hat \theta(\vec x, \vec u)]$,
and by R3 the $\vec u$-part of $\beta$ must be strictly smaller than the $\vec y$-part of~$\alpha$.
By the construction of~$\omega$, if we let $\alpha' := (\alpha \cup \beta) \circ \omega$
then $\alpha'$ again satisfies $A(\vec x) \wedge [\hat \theta(\vec x, \vec y)]$,
with the $\vec y$-part of $\alpha'$ the same as the $\vec u$-part of $\beta$.
In this way $\Delta$ and $\omega$ work together to simulate one step in the 
exponential-time algorithm to solve PLS by producing smaller and smaller 
``feasible solutions'' $\vec y$ such that $[\hat \theta(\vec x, \vec y)]$. Specifically,
$\Delta$ computes the next solution and writes it on its new variables~$\vec u$, 
then $\omega$ copies the values of $\vec u$ back to the old variables, overwriting $\vec y$.

Formally, to complete the proof we need to construct two ER derivations
\begin{enumerate}
\item[(a)]
$\Gamma \wedge \Delta \vdash \Gamma_{\restrict \omega}$
\item[(b)]
$\Gamma \wedge \Delta  \wedge [ \vec v \lelex \vec v_{\restrict \omega}] \vdash \bot.$
\end{enumerate}
Strictly speaking we should also include $\neg C$ in both sets of assumptions, but
since $C$ is the empty clause omitting this makes no difference.

Writing out (a) in more detail, what we need to show is
\[
A(\vec x) \wedge [\hat \theta(\vec x, \vec y)]
 \wedge 
[(\vec u, \vec z_N) {=} \hat N (\vec x, \vec y)] \wedge [\vec z_u {=} D_\theta (\vec x, \vec u)]
\ \vdash \
A(\vec x) \wedge [\hat \theta(\vec x, \vec u)].
\]
If the clauses $[\vec z_u = D_\theta (\vec x, \vec u)]$ were not present on the left then $R2$ would
already be a derivation of this. However these clauses are part of $[\hat \theta(\vec x, \vec u)]$, 
so we can use $R2$ with an appeal to Lemma~\ref{lem:ER_add_conclusions_to_assumptions}.

For (b), we observe that $[ \vec v \lelex \vec v_{\restrict \omega}]$
means precisely $[ (\vec y, \vec z_y, \vec x) \lelex (\vec u, \vec z_u, \vec x) ]$,
from which formula we can easily derive in ER that $[ \vec y \lelex \vec u]$.
Hence we can use $R3$ and Lemma~\ref{lem:compose_ER_derivations}.
\end{proof}

\section{Linear dominance simulates $\ERpls$} \label{sec:lin-dom>er-pls}

Recall from Section~\ref{sec:CP} the notation $C^*$ and $\Gamma^*$ for converting clauses and CNFs
into equivalent PB constraints and formulas.
We will show that, given a derivation of $\Delta$ from $\Gamma$
in $\ERpls$, we can construct in polynomial time a derivation of $(\Delta^*, \emptyset, \top, \emptyset)$ 
from $(\Gamma^*, \emptyset, \top, \emptyset)$ 
in the linear dominance system, which implies the simulation. 
So we must show how to handle the two rules of $\ERpls$: the ER rule and the dominance rule.

\paragraph{ER rule.} 
This follows straightforwardly by the well-known simulation of resolution by cutting planes~\cite{cook1987complexity}
and, for extension steps, using the redundance-based strengthening
rule (of the linear dominance system) and standard arguments about how to add
extension axioms as redundant clauses, see e.g.~\cite{kullmann1999generalization}.
We include a detailed proof in Appendix~\ref{appendix_simulate_ER}.

\paragraph{Dominance rule.}
Suppose we have a CNF $\Gamma$ and a clause $C$, both in variables $x_1, \dots, x_n$, plus
a set $\Delta$ of extension axioms over $\vec x ; \vec y$,
a substitution $\omega$ mapping variables~$\vec x$ to variables $\vec x \cup \vec y$,
and two $\ER$ derivations
\begin{enumerate}
\item[(a)]
$\Gamma \wedge \Delta \wedge  \neg C \vdash \Gamma_{\restrict \omega}$
\item[(b)]
$\Gamma \wedge \Delta \wedge  \neg C \wedge [ \vec x \lelex \vec x_{\restrict \omega}] \vdash \bot.$
\end{enumerate}
We will describe a derivation from $(\Gamma^*, \emptyset, \top, \emptyset)$
of $(\Gamma^* \cup \{ C^* \}, \emptyset, \top, \emptyset)$.

We can combine the derivations for~(a) and~(b)
above to construct an ER derivation 
\begin{equation*} \label{eq:res_strict}
\pi : \Gamma \wedge \Delta \wedge  \neg C \vdash 
	\Gamma_{\restrict \omega} \wedge [ \vec x_{\restrict \omega} \llex \vec x].
\end{equation*}
where $[\vec x \llex \vec y]$ represents strict lexicographic ordering 
-- see Proposition~\ref{pro:ER_moving_ordering} in Appendix~\ref{sec:ordering}.
Furthermore we may assume that $\pi$ is actually a \emph{resolution} derivation, that is, 
that it does not include any applications of the extension rule. 
This is because we can move all extension axioms introduced by that rule
from the body of the derivation to $\Delta$, preserving the order in which they appeared
in the derivation. 
That process turns $\Delta$ into a set of extension axioms over $\vec x ; \vec y, \vec z$,
where $\vec z$ now includes all extension variables that were introduced
in the original~$\pi$, and in particular all auxiliary variables in~$[ \vec x_{\restrict \omega} \llex \vec x]$.

To build the derivation in the linear dominance system,
we first change the ordering from the trivial order $\top$ to the lexicographic order
on $x_1, \dots, x_n$, with the most significant bits first. This
can be done 
using the order change rule;
see Section~\ref{sec:linear_dominance_system}.
So we are now in the configuration $(\Gamma^*, \emptyset, \OOl, \vec x)$,
where $\OOl$ is the lexicographic order.

We then derive $(\Gamma^*, \Delta^*, \OOl, \vec x)$, where we add each 
extension axiom in $\Delta$ in turn using the redundance-based strengthening rule, 
in the same way that we handle
extension axioms in the the ER rule. 
We must check that we satisfy the order condition for this rule,
but this is easy, since the substitutions used do not affect $\vec x$ variables,
which are the only variables relevant to the ordering.
Again we have moved the details to Appendix~\ref{appendix_simulate_ER}.

Now we use dominance-based strengthening
to derive $(\Gamma^*, \Delta^* \cup \{  C^* \}, \OOl, \vec x)$.
We apply the normal translation from resolution into CP to $\pi$
to get 
\[
\Gamma^* \cup \Delta^* \cup  \{(\neg C)^*\} \vCP 
	(\Gamma_{\restrict \omega})^* \cup [ \vec x_{\restrict \omega} \llex \vec x]^*.
\]
It is easy to construct a short derivation $\neg (C^*) \vCP (\neg C)^*$.
We can also construct a derivation
 $[ \vec x_{\restrict \omega} \llex \vec x]^* \vCP \OO_\prec(\vec x_{\restrict \omega}, \vec x)$
in polynomial time by Lemma~\ref{lem:ordering_derivation} 
in Appendix~\ref{sec:ordering},
where $\OO_\prec$ is strict lexicographic ordering written in the natural 
way using the same multilinear 
function $f$ as~$\OOl$;
this is the same as $L_<$ from Lemma~\ref{lem:ordering_derivation}, except that $L_<$ is
reverse-lexicographic.
Moreover,  $(\Gamma_{\restrict \omega})^*$
is the same as $(\Gamma^*)_{\restrict \omega}$.
Thus we have
\begin{equation} \label{eq:CP_strict}
\Gamma^* \cup \Delta^* \cup  \{\neg (C^*) \} \vCP 
	(\Gamma^*)_{\restrict \omega} \cup \OO_\prec (\vec x_{\restrict \omega}, \vec x).
\end{equation}
Finally, from (\ref{eq:CP_strict}) we can trivially construct a derivation
\begin{equation} \label{eq:CP_non_strict}
\Gamma^* \cup \Delta^* \cup  \{\neg (C^*) \}
	\cup \OOl( \vec x, \vec x_{\restrict \omega}) 
	\vCP \bot.
\end{equation}
The derivations (\ref{eq:CP_strict}) and (\ref{eq:CP_non_strict}) are what
we need to apply the dominance-based strengthening rule to derive $C^*$
(after weakening $\OO_\prec$ in (\ref{eq:CP_strict}) to $\OOl$).
This completes the proof.

\section{$G_1$ simulates linear dominance}\label{sec:soundness}

For this result we will use Theorem \ref{thm:t12-g1}, which states
that for a propositional proof system $Q$, if $S^2_2$ proves the
CNF-reflection principle for $Q$, then $G_1$ simulates $Q$.
We take $Q$ to be the linear dominance system, considered as a system
for refuting CNFs. That is, a $Q$-refutation of a CNF $\Gamma$ is a
linear dominance refutation of $\Gamma^*$.
Thus for the simulation it is enough to prove in $S^2_2$ that the existence
of such a refutation of $\Gamma^*$ implies that $\Gamma$ is unsatisfiable.

We do this by formalizing in $S^2_2$ as much as we can of the proof of soundness
of the dominance system from~\cite{bgmn:journal}.
We run into a problem when dealing with the dominance rule.
To show it is sound, it is enough to show that if a CNF is satisfiable, then 
it has a least satisfying assignment with respect to the ordering $\OOl$.
However as far as we know the general statement of this form, that an
arbitrary ordering has a least element,
is not provable in $S^2_2$, and is known
 to be unprovable if the ordering given is by an oracle~\cite{Riis_1994}.
It \emph{is} provable in $T^2_2$, and hence in $S^3_2$, by a simple inductive argument.
By the methods in this section it follows from this that the full dominance
system is simulated by~$G_2$.

To stay inside the strength of $S^2_2$ we chose to work with the linear
dominance system instead since $T^1_2$, and hence also $S^2_2$,
can prove that any nonempty set of strings has a least element in the lexicographic
ordering, which is enough to prove the soundness of the dominance
rule restricted to such an ordering.\footnote{%
Whether the dominance system is strictly stronger than the linear dominance
system is unclear. Conceivably to bound the strength of the full dominance rule
we could make use of the fact that the ordering is not given by an arbitrary
relation, but by a relation which is provably an ordering in $\CP$.} 

We use a definition from~\cite{bgmn:journal}.

\begin{definition} \label{def:validity}
A configuration $( \CC, \DD, \OOl, \vec z )$ is called \emph{valid} if
\begin{enumerate}
\item
$\CC$ is satisfiable
\item
For every total assignment $\alpha \vDash \CC$, there is
a total assignment $\beta$ with $\beta \preceq \alpha$ and $\beta \vDash \CC \cup \DD$.
\end{enumerate}
\end{definition}

Notice that validity is a $\Pi^b_2$ condition, and in particular $S^2_2$ is strong
enough to do all the basic reasoning we need about sums and inequalities.
Working in $S^2_2$,
suppose for a contradiction we are given a satisfiable CNF $\Gamma$
and a linear dominance refutation $\pi$ of $\Gamma^*$. We will
use $\Pi^b_2$ length induction, taking as our
inductive hypothesis that the $i$th configuration in $\pi$ is valid.
The base case is the initial configuration 
$(\Gamma^*, \emptyset, \top, \emptyset)$, which is
valid by assumption (since any assignment satisfying $\Gamma$ already satisfies $\Gamma^*$). 
On the other hand the final configuration in $\pi$ is not valid,
since in that configuration $\CC \cup \DD$ is not satisfiable.
Therefore to derive a contradiction it is enough to show that every rule preserves validity.

Since the soundness of CP is trivially provable in $\PV$,
 this is easy for the implicational derivation,
transfer and order change rules. For the remaining rules, 
namely redundance-based strengthening, deletion and dominance-based strengthening,
suppose we are at a valid configuration
$( \CC, \DD, \OOl, \vec z )$.

\paragraph{Redundance-based strengthening rule.}
We have a substitution $\omega$ and we know 
\[
\CC \cup \DD \cup \{ \neg C \}
\vDash
(\CC \cup \DD \cup \{ C \})_{\restrict \omega} \cup \OOl (\vec z_{\restrict \omega}, \vec z).
\]
Let $\alpha \vDash \CC$. By the inductive hypothesis  there is $\beta \preceq \alpha$ such that
$\beta \vDash \CC \cup \DD$, and we want to find $\beta' \preceq \alpha$ such that
$\beta' \vDash \CC \cup \DD \cup \{ C \}$.
If $\beta \vDash C$ then we set $\beta' = \beta$. Otherwise we set $\beta' = \beta \circ \omega$,
and the properties of $\beta'$ follow from the assumption.

\paragraph{Deletion rule.}
The interesting case is that we derive
$(\CC', \DD', \OOl, \vec z)$ with $\DD' \subseteq \DD$
and $\CC' = \CC \setminus \{ C \}$ for some $C$
derivable by the redundance rule from $(\CC', \emptyset, \OOl, \vec z)$.
Let $\alpha \vDash \CC'$. If $\alpha \vDash C$ then there is nothing to show.
Otherwise, using the notation of the redundance rule, we let $\alpha'=\alpha \circ \omega$
and know that $\alpha' \preceq \alpha$ and $\alpha' \vDash \CC$.
The inductive hypothesis then gives us $\beta' \preceq \alpha'$ with
$\beta' \vDash \CC \cup \DD$, so in particular
$\beta' \preceq \alpha$ and $\beta' \vDash \CC' \cup \DD'$.

\medskip

For all rules so far, the proof that validity is preserved goes through even in $\PV$.
For the last rule we will need to minimize the value of a $\PV$ function on a polynomial-time computable set,
which can be done in $S^2_2$, since it extends $T^1_2$ 
(see Section \ref{subsec:ba}). 

\paragraph{Dominance-based strengthening rule.}
We derive 
$(\CC, \DD \cup \{ C \}, \OOl, \vec z)$, and for  
 a given substitution $\omega$ we know
 \begin{align*}
 \CC \cup \DD \cup \{ \neg C \}
& \vDash
\CC_{\restrict \omega} \cup f(\vec z_{\restrict \omega}) \le f(\vec z) \\
\CC \cup \DD \cup \{ \neg C \} \cup \{ f(\vec z) \le f(\vec z_{\restrict \omega})\} & \vDash \bot
\end{align*}
where $f$ is the linear function defining $\OOl$.

Let $\alpha \vDash \CC$. We want to find $\beta' \preceq \alpha$ such that
$\beta' \vDash \CC \cup \DD \cup \{ C \}$.
Let $S$ be the set of total assignments $\preceq$-below $\alpha$ satisfying $\CC$.
Let $\beta$ be a member of $S$ for which
$f(\beta)$ is minimal (where $f(\beta)$ stands for $f$ applied to the
$\vec z$ variables of $\beta$). Using the least number principle for $\Sigma^b_1$ formulas
available in $S^2_2$, we can find such a $\beta$.

By the inductive hypothesis (that is, the validity of $(\CC, \DD , \OOl, \vec z)$)
we may assume  that $\beta \vDash \CC \cup \DD$.
If $\beta \vDash C$ then we set $\beta' = \beta$.
Otherwise let $\beta' = \beta \circ \omega$.
By the first entailment in the rule, $\beta' \vDash \CC$ and $f(\beta') \le f(\beta)$,
so $\beta' \in S$.
Therefore by the minimality 
of $f(\beta)$ we have $f(\beta) \le f(\beta')$. 
But this contradicts the second entailment.

\medskip

This completes the proof that $S^2_2$ proves the soundness of the linear dominance
system, which is thus simulated by $G_1$.

\section{Simulations of fragments by $\ER$} \label{sec:upper}

\subsection{Weak linear dominance}\label{subsec:weak-dominance}

Consider the version of the linear dominance system in which we limit
 the dominance-based strengthening rule by only allowing it to be applied
  when the set $\DD$ of derived clauses is empty.
That is, we replace it with the rule:
from $(\CC, \emptyset, \OOl, \vec z)$ derive
$(\CC \cup \{ C \},  \emptyset, \OOl, \vec z)$
if there is a substitution $\omega$ and derivations 
\begin{align*}
 \CC \cup \{ \neg C \}
& \vCP
\CC_{\restrict \omega} \cup f(\vec z_{\restrict \omega}) \le f(\vec z) \\
\CC \cup \{ \neg C \} \cup \{ f(\vec z) \le f(\vec z_{\restrict \omega})\} & \vCP \bot.
\end{align*}
where $f$ is the linear function defining $\OOl$. We shall refer to this system as the \emph{weak linear dominance system}.

\begin{prop} \label{pro:ER_simulation}
The weak linear dominance system is simulated by ER.
\end{prop}

To prove the simulation, we use a lemma saying that $\PV$ knows
that there is a polynomial time function that lets us iterate a substitution
$m$ times, when $m$ is given in \emph{binary}.

\begin{lemma} \label{lem:iterate_omega}
There is a polynomial time function $g(\omega, m)$ which takes as input
a substitution~$\omega$ on variables $z_1, \dots, z_n$  and a number $m$ (coded in binary)
and outputs the substitution~$\omega^m$.
Furthermore this works provably in $\PV$, that is, 
$\PV \vdash g(\omega, m+1) =  g(\omega, m) \circ \omega$.
\end{lemma}

\begin{proof}
Fix a variable $z_i$ and consider the sequence
$z_i, \omega(z_i), \dots, \omega^{2n+2}(z_i)$ as a
walk through the space $\Lit \cup \{ 0,1\}$.
The sequence can be produced by a $\PV$ function on input $\ang{\omega,i}$.
By the pigeonhole principle, which is available in $\PV$ here since $n$ is 
small (polynomial in the length of the input),
this sequence must touch some point twice.
That is, it consists of a walk of length $k_i$ to some $u \in \Lit \cup \{0,1\}$,
followed by a loop of some size $\ell_i$, where $0 \le k_i \le 2n + 1$ and
$1 \le \ell_i \le 2n$. Again, the numbers $k_i, \ell_i$ can be computed
by a $\PV$ function on input~$\ang{\omega,i}$

Thus to compute $\omega^m(z_i)$ for $m>2n+2$ it is enough to calculate
the remainder of $m-k$ divided by $\ell$: namely, 
$\omega^m(z_i) = \omega^{k_i + ((m - k_i) \bmod \ell_i)}(z_i)$. 
\end{proof}

\begin{proof}[Proof of Proposition \ref{pro:ER_simulation}]

By Theorem \ref{thm:t12-g1}, 
it is enough to show that the soundness of 
the weak linear dominance system is provable in $S^1_2$. 
So, working in $S^1_2$, suppose that a CNF $\Gamma$ is satisfiable 
but that $\Gamma^*$ has a 
refutation $\pi$ in the system. We will derive a contradiction. We will use length induction, 
but with a weaker inductive hypothesis than was used in Section~\ref{sec:soundness}
for the soundness of full linear dominance.
Namely, we will show that for each configuration
$(\CC, \DD, \OOl, \vec z)$ in turn in $\pi$, $\CC \cup \DD$ is satisfiable.
Satisfiability is a $\Sigma^b_1$ property, so this is a form of length
induction we can carry out in $S^1_2$.
It yields a contradiction when we get to the last configuration in $\pi$.

The first configuration is satisfiable, by the assumption on $\Gamma$.
It is easy to see that every rule, other than dominance-based
strengthening, preserves satisfiability; in the case of the redundance-based
strengthening rule, this is by the standard argument about composing the
current assignment once with $\omega$, if necessary.

So suppose we are dealing with the weak dominance-based strengthening rule.
We have an assignment $\alpha$ which satisfies the current configuration
 $(\CC, \emptyset, \OOl, \vec z)$, and we want to satisfy
$(\CC \cup \{ C \},  \emptyset, \OOl, \vec z)$.
We have a substitution~$\omega$ and derivations
\begin{align*}
 \CC \cup \{ \neg C \}
& \vCP
\CC_{\restrict \omega} \cup f(\vec z_{\restrict \omega}) \le f(\vec z) \\
\CC \cup \{ \neg C \} \cup \{ f(\vec z) \le f(\vec z_{\restrict \omega})\} & \vCP \bot.
\end{align*}
for a linear $f$. Suppose for a contradiction 
that $\CC \cup \{ \neg C \}$ is unsatisfiable. Then, since CP derivations are provably
sound (even in $\PV$) we know that for any assignment $\beta$,
if $\beta \vDash \CC$, then $\beta \circ \omega \vDash \CC$ and $f(\beta \circ \omega)<f(\beta)$.

We may assume without loss of generality that $f$ only takes values between~$0$ and
some upper bound~$m$. Writing $\alpha_i$ for $\alpha \circ \omega^i$, we
use induction (rather than length induction) on~$i$  to show that for all~$i$ we have
\[
\alpha_i \vDash \CC \ \textrm{~and~} \ f(\alpha_i) \le m-i.
\]
By Lemma \ref{lem:iterate_omega}, this is a 
$\PV$ formula, 
so this induction can be carried out
in~$S^1_2$ (if the formula were $\Sigma^b_1$, we would only be able to use length induction).
The base case $i=0$ is true by the assumptions about $\alpha$ and $f$, and the 
inductive step follows from the discussion in the previous paragraphs. We conclude that
$f(\alpha_{m+1}) \le -1$, which is impossible.\end{proof}

\subsection{Symmetry breaking in $\ER$} \label{sec:ER_breaking}

Let us  define a proof system $Q$, which we could  call
$\ER$ plus \emph{static symmetry breaking}. A refutation of a CNF $\Gamma$ in $Q$
consists of an initial step, in which we list a sequence of symmetries 
$\omega_1, \dots, \omega_k$ of $\Gamma$ and write down the corresponding 
lex-leader constraints (where for each constraint we use fresh auxiliary variables). 
This is followed by an $\ER$ refutation of $\Gamma$ augmented by these constraints,
that is, of $\Gamma' := \Gamma \wedge \bigwedge_i [\vec z \lelex \vec z_{\restrict \omega_i}]$.

For $k \in \NN$ we define $Q_k$ to be $Q$ limited to only adding axioms for $k$~symmetries.

\begin{prop}
The full system $Q$ is sound, and is simulated by $G_1$.
\end{prop}

\begin{proof}
We repeat the proof of Proposition~\ref{pro:intro_symmetry},
except this time we fill in some details.
To prove soundness, it is enough to show that,
supposing $\Gamma$ is satisfiable, 
$\Gamma'$ is satisfiable as well. 
Let~$\alpha$ be a lexicographically
 minimal assignment to the $\vec z$-variables satisfying $\Gamma$. 
We claim that an extension of $\alpha$ satisfies~$\Gamma'$.
To see this, let $\omega_i$ be any symmetry from our list.
Then $\alpha \vDash \Gamma$ implies $\alpha \vDash \Gamma_{\restrict \omega_i}$,
and thus $\alpha \circ \omega_i \vDash \Gamma$.
By minimality of $\alpha$ we have $\alpha \lelex \alpha \circ \omega_i$,
and thus, extending $\alpha$ to $\beta$ which satisfies the extension axioms
in the definition of $\lelex$, we have that $\beta$ satisfies the symmetry-breaking axiom 
$[\vec z \lelex \vec z_{\restrict \omega_i}]$. In this way we can simultaneously
satisfy such axioms for all $i$, by the assumption that auxiliary variables are disjoint.

For the simulation by $G_1$, it is enough to observe that this argument
can be cast as a proof of the CNF-reflection principle for $Q$ and carried out in $T^1_2$.
Then we can
appeal to Theorem~\ref{thm:t12-g1}.
\end{proof}

The converse direction is presumably false:

\begin{prop}
$G_1$ is not simulated by $Q$, assuming $G_1$ is not simulated by $\ER$.
\end{prop}

\begin{proof}
Let $\Gamma_n$ be a family of CNFs which have polynomial-sized refutations
in $G_1$ but require superpolynomial size in $\ER$. 
Then it is easy to construct a polynomial-sized CNF $A_n$ such that $\Gamma_n \cup A_n$ 
has no symmetries; 
assuming $\Gamma_n$ has variables $x_1, \dots, x_m$,
a convenient example consists of clauses $x_i \vee y_1 \vee \dots \vee y_i$
for  each~$i$, where $y_1, \dots, y_m$ are new variables.
Then $G_1$ refutations of $\Gamma_n$ still work for  
$\Gamma_n \cup A_n$ (we may need to add one more weakening step).
On the other hand, if $\pi$ is any $Q$ refutation of $\Gamma_n \cup A_n$, 
then it must be just an $\ER$ refutation, and we can turn it into an $\ER$
refutation of $\Gamma_n$ by applying the restriction which sets every
$y_i$ variable to $1$. Thus $\pi$ must have superpolynomial size.
\end{proof}

We can now prove Theorem~\ref{the:intro_Q1} from the introduction,
that $Q_1$ is simulated by $\ER$.

\begin{proof}[Proof of Theorem~\ref{the:intro_Q1}]
We will show that the soundness of $Q_1$ is provable in $S^1_2$. The result then 
follows by Theorem~\ref{thm:t12-g1}. 
Let $\Gamma$ be a CNF and let $\omega$ be a symmetry of $\Gamma$. 
Let $\Gamma' := \Gamma \wedge [\vec z \lelex \vec z_{\restrict \omega}]$
and suppose we are given an $\ER$ refutation of $\Gamma'$.
We will show, with a proof formalizable
in $S^1_2$, that if $\Gamma$ is satisfiable then so is 
$\Gamma'$. We can then derive a contradiction, since
$S^1_2$ proves the soundness of $\ER$ (see Theorem~\ref{the:provable_G1_soundness}).

Working in $S^1_2$,
suppose $\alpha \vDash \Gamma$. 
As in the proof of Proposition~\ref{pro:ER_simulation}
we write $\alpha_i$ for $\alpha \circ \omega^i$, and use 
the fact that by Lemma \ref{lem:iterate_omega} this can be computed in polynomial time.
Suppose for a contradiction that $\Gamma \wedge [\vec z \lelex \vec z_{\restrict \omega}]$
is unsatisfiable. It follows that for any assignment $\beta$, if $\beta \vDash \Gamma$
then $\beta \circ \omega \llex \beta$. On the other hand, if $\beta \vDash \Gamma$ then 
we already know $\beta \circ \omega \vDash \Gamma$, since $\Gamma_{\restrict \omega} = \Gamma$.
Assuming that there are $n$ many $z$-variables we have $\alpha \lelex 2^n-1$,
where we identify $2^n-1$ with a string of $1$s of length~$n$. 
Thus we can reach a contradiction by a similar induction as in 
the proof of Proposition~\ref{pro:ER_simulation}, showing 
inductively that for each~$i$ we have $\alpha_i \vDash \Gamma$ and $\alpha_i \lelex 2^n-i$.
\end{proof}

This proof breaks down immediately even for $Q_2$, since we do not have any
equivalent of Lemma~\ref{lem:iterate_omega}
for arbitrary compositions of two substitutions.

We briefly discuss how one could directly construct an ER refutation from a $Q_1$
refutation, without going through bounded arithmetic and Theorem~\ref{the:intro_Q1}.
The main task is to construct a circuit $C$ which, when given an assignment $\alpha$
such that $\alpha \vDash \Gamma$, outputs an assignment $\beta$ such
that $\beta \vDash \Gamma \wedge [\vec z \lelex \vec z_{\restrict \omega}]$.
Furthermore this property of $C$ must be provable in ER, in the sense that we have an ER derivation
$\Gamma (\vec x) \wedge [\vec z = C(\vec x)] \vdash \Gamma(\vec z) \wedge [\vec z \lelex \vec z_{\restrict \omega}]$
(where we are suppressing auxiliary variables in $C$ and $\lelex$). We will just describe~$C$.

We use a subcircuit which takes input $\vec z, i$ and computes $\alpha_i := \vec z_{\restrict \omega^i}$
using the algorithm for~$g$ in Lemma~\ref{lem:iterate_omega}.
The circuit $C$  finds~$i$ such that the two conditions
$\alpha_i \vDash \Gamma$  and
$\alpha_i \lelex 2^n - i$
hold for $i$, but one of them fails for~$i+1$,
and outputs $\alpha_i$. Such an~$i$ can be found by binary
search, since both conditions hold for $i=0$ and the second one 
must fail for~$i=2^n+1$.
Since $\alpha_i \vDash \Gamma$ and $\alpha_{i+1} = \alpha_i \circ \omega$,
we have that $\alpha_{i+1} \vDash \Gamma$ as $\omega$ is a symmetry.
We conclude that the second condition fails and $\alpha_{i+1} >_\mathrm{lex} 2^n-i-1$. 
Thus $\alpha_{i+1} \ge_\mathrm{lex} \alpha_i$, meaning that $\alpha_i \lelex \alpha_i \circ \omega$
as required.

\paragraph{Acknowledgements.}
We are grateful to Jakob Nordstr\"{o}m
for introducing us to this topic and answering our questions about it, 
and to Sam Buss and Vijay Ganesh for other helpful discussions.

\bibliography{dominance_aib}

\appendix


\section{Postponed technical material} \label{sec:technical_appendix}


\subsection{How to formalize ordering (from Section~\ref{sec:ERpls})} 
\label{sec:ordering}

We 
define CNFs $[ \vec x \llex \vec y]$ and $[ \vec x \lelex \vec y]$ expressing the lexicographic ordering.
We will need these to be compatible with how we reason about
orderings in the pseudo-Boolean setting.

Let $r$ be the arity of $\vec x$ and $\vec y$.
For our proof system we will want to compare 
 $\vec x$ and $\vec y$  with the 
most significant bits first, but just in this section we prefer to define a CNF for reverse lexicographic
order, since it gives more readable notation when we spell out the details.

We work with variables $x_1, \dots, x_r$, $y_1, \dots, y_r$, $z_1, \dots, z_{r+1}$
and $c_2, \dots , c_{r+1}$, representing numbers $x ,y< 2^r$, a number $z<2^{r+1}$ and a string
of ``carry'' bits~$c$.
We first write a CNF~$\Delta$ expressing ``$x+z = y + 2^r$''. 
This is the conjunction of, 
 for each~$i=1, \dots, r+1$, a CNF expressing
\begin{equation}\label{eqn:bit-computation}
x_i+z_i+c_i =y_i + 2c_{i+1}
\end{equation}
where at the ``boundaries'' we replace $x_{r+1}$, $c_1$ and $c_{r+2}$ with the constant~$0$ and
replace $y_{r+1}$ with the constant~$1$; this last step corresponds to adding $2^r$
to the number $y$. 
Each expression of the form (\ref{eqn:bit-computation}) can be written as a CNF.
Moreover, since the number $z$ is uniquely defined, 
we can introduce extension variables needed for intermediate steps in the verification of (\ref{eqn:bit-computation})
in such a way
that the whole formula $\Delta$ becomes a set of
extension axioms defining $\vec z$, $\vec c$ and the intermediate extension variables
from $\vec x$ and $\vec y$.


We  write a CNF $P_<$ expressing that $\vec x < \vec y$
by saying that ``$x+z=y+2^r$ with $z > 2^r$'',
that is, $\Delta \wedge z_{r+1} \wedge \bigvee_{i=1}^r z_i$.

\medskip

Now let  $L_<$ be the natural way to expressing the strict lexicographic ordering
 by a single PB constraint. That is, $L_<$ is the constraint
$\sum_{i=1}^r 2^{i-1} x_i < \sum_{i=1}^r 2^{i-1} y_i$.
The definition of $P_<$ was chosen to give 
a simple proof of the following lemma.

\begin{lemma} \label{lem:ordering_derivation}
There is a polynomial-time procedure that on input $1^r$ produces
CP derivations $P_<(\vec x, \vec y)^* \vCP L_<(\vec x, \vec y)$.
\end{lemma}

\begin{proof}
We will write $X_j, Y_j, Z_j$ for the sums
$\sum_{i=1}^j 2^{i-1} x_i$ etc. We will inductively construct
derivations
$
\Delta^* \vCP X_j + Z_j = Y_j + 2^{j} c_{j+1},
$
for $j =  1, \dots, r$, where $\Delta$ is the CNF defined above. 
Here and below the equality $=$ is
shorthand for the pair of inequalities~$\le$ and~$\ge$.

For each such $j$, the CNF $\Delta$ contains clauses together
expressing $x_i+z_i+c_i =y_i + 2c_{i+1}$,
after the substitution $c_1 \rightarrow 0$.
Thus, by the implicational completeness
of cutting planes~\cite{chvatal:edmonds-polytopes}, we can produce constant-size derivations
of this equality (as a pair of PB constraints)
from $\Delta^*$.
For $i=1$ this immediately gives the base case of the induction.
Now, in the inductive step, suppose $1 \le j < r$
and we have  $X_j + Z_j = Y_j + 2^{j} c_{j+1}$.
Then we add $x_{j+1} + z_{j+1} = y_{j+1} + 2 c_{j+2} - c_{j+1}$ multiplied by $2^j$ 
to obtain $X_{j+1} + Z_{j+1} = Y_{j+1} + 2^{j+1} c_{j+2}$.

When the inductive construction is completed, 
we have derived $X_r + Z_r = Y_r + 2^{r} c_{r+1}$.
The final constraints in $\Delta^*$ can be used to produce a
constant-size derivation of the equality $z_{r+1} + c_{r+1} = 1$.
So we can add $2^r z_{r+1} = 2^r - 2^r c_{r+1}$ to get
$X_r + Z_{r+1} = Y_r + 2^r$.
From the assumptions $z_{r+1}=1$ and $\bigvee_{i=1}^r z_r$
in $P_<$ we can derive $Z_{r+1} > 2^r$, and together these give us $L_<$.
\end{proof}

We define $P_\le(\vec x, \vec y)$ as the negation of $P_<(\vec y, \vec x)$.
Precisely, if $P_<(\vec y, \vec x)$ has the form $\Delta \wedge z_{r+1} \wedge \bigvee_{i=1}^r z_i$
we take $P_\le(\vec x, \vec y)$ to be the CNF $\Delta \wedge [ w \leftrightarrow (z_{r+1} \wedge \bigvee_{i=1}^r z_i) ]
\wedge \neg w$, where the expression $[ w \leftrightarrow \dots ]$ is shorthand for a series of extension axioms
defining a new variable~$w$.
Finally we set the formula
$[\vec x \lelex \vec y]$ to be $P_\le(x_r, \dots, x_1, y_r, \dots , y_1)$ 
and set
$[\vec x <_\mathrm{lex} \vec y]$ to be $P_<(x_r, \dots, x_1, y_r, \dots , y_1)$.
The next proposition is proved in a similar way to the proof of Proposition~\ref{pro:PV_to_ER}
below.

\begin{prop} \label{pro:ER_moving_ordering}
Given an $\ER$ derivation $\Gamma \wedge [ \vec x \lelex \vec y ] \vdash \bot$
we can construct in polynomial time an $\ER$ derivation
$\Gamma \vdash [ \vec y <_\mathrm{lex} \vec x ]$.
\end{prop}

\subsection{Propositional translations (from Section~\ref{subsec:prop-transl})} 
\label{appendix_translations}

We give the proof of Proposition~\ref{pro:PV_to_ER}.
That is, suppose $\PV$ proves 
a sentence 
\begin{equation}\label{eqn:provable-formula}
\forall \vec x, \, \phi_1(\vec x) \wedge \dots \wedge \phi_r(\vec x)
\rightarrow \theta(\vec x),
\end{equation}
where $\phi_1, \dots, \phi_r, \theta$ are quantifier-free. 
Then for any assignment $\vec k$ of bit-lengths to the variables $\vec x$,
we can construct in time polynomial in $\vec k$
an $\ER$ derivation 
\[
\norm{\phi_1(\vec x)}_{\vec k} \wedge \dots \wedge \norm{\phi_r(\vec x)}_{\vec k}
\vdash \norm{ \theta(\vec x)}_{\vec k}.
\]

\begin{proof}[Proof of Proposition~\ref{pro:PV_to_ER}]
By the usual form of the translation from $\PV$ into $\ER$,
as presented for example in~\cite[Theorem 12.4.2]{krajicek:proof-complexity},
it follows from the provability of (\ref{eqn:provable-formula})
that in time polynomial in $\vec k$ we can build 
an $\ER$ \emph{refutation} 
\[
\pi :
\norm{\phi_1(\vec x)}_{\vec k} \wedge \dots \wedge \norm{\phi_r(\vec x)}_{\vec k} \wedge \norm{\neg \theta(\vec x)}_{\vec k} \vdash \bot.
\]
The details of how to transform $\pi$ into an $\ER$ derivation 
$\pi' : \Gamma \vdash \norm{ \theta(\vec x)}_{\vec k}$,
where we write $\Gamma$ for $\norm{\phi_1(\vec x)}_{\vec k} \wedge \dots \wedge \norm{\phi_r(\vec x)}_{\vec k}$, 
are routine but messy. 
The translation 
 $\norm{\theta(\vec x)}_{\vec k}$ 
has the form $\Delta \wedge z^\out$ where
$\Delta$ is a set of extension axioms over $\vec x_1, \dots, \vec x_\ell ; \vec z$ for some variables~$\vec z$
(which include~$z^\out$).
Modulo trivial
modifications, $\norm{\neg\theta(\vec x)}_{\vec k}$ can  be taken to be $\Delta \wedge \neg z^\out$.
Note that by our conventions on variables, the variables $\vec z$ do not appear in~$\Gamma$.
Thus we can begin $\pi'$ by deriving all 
of $\Delta$ from $\Gamma$ using the extension rule, and all that remains
is to derive $z^\out$.

To so this,
 we work through the refutation $\pi : \Gamma \wedge \Delta \wedge \neg z^\out \vdash \bot$
and, for each clause~$C$ in $\pi$ except for the initial clause $\neg z^\out$,
we derive $C \vee z^\out$. 
If $C$ is a clause of $\Gamma$ or $\Delta$, we derive $C \vee z^\out$ from $C$ by weakening
(notice that $z^\out$ is at this point an ``old'' variable, since it appears in $\Delta$);
if $C$ is derived in $\pi$ by resolution on a variable other than $z^\out$
or by weakening, we derive $C \vee z^\out$ by the same inference; 
if $C$ is derived in $\pi$ from $B \vee z^\out$ and $A \vee \neg z^\out$ by resolution on~$z^\out$, we
instead derive $C \vee z^\out$ by weakening from $B \vee z^\out$ (again, this does not introduce any new variables).
If $C$ is part of an extension rule  in $\pi$, we introduce the same extension axioms and then
derive $C \vee z^\out$ by weakening, noting that the same variables are new at this point in $\pi$ and in $\pi'$.
Thus we finally derive~$z^\out$, where in $\pi$ we derived~$\bot$.
\end{proof}

\subsection{How to formalize satisfiability (from Section~\ref{sec:ERpls_simulates_G1})} \label{sec:formalize_sat}

\begin{prop} \label{pro:A_to_phi}
There is a polytime procedure that, 
given a CNF $A$ in $n$ variables, produces an $\ER$
derivation $A(\vec x) \vdash \norm{\Sat(a,x)}_{\restrict \tau}$,
where the bit-lengths of $\vec a, \vec x$ in $\norm{\Sat(a,x)}$ are the size of $A$ and $n$, respectively,
$\tau$ substitutes the bits of $A$ for $\vec a$, 
and $A(\vec x)$ is $A$ with variables renamed to $\vec x$.

\end{prop}

\begin{proof}
We may assume that $A$ is already in the variables $\vec x = x_1,\ldots,x_n$.
The details of the procedure will depend on how exactly $\Sat(a,x)$ is formalized
by a $\PV$ formula, but any two reasonable formalizations will
be provably equivalent in $\PV$, so we may focus on one convenient
formalization and then invoke Proposition \ref{pro:PV_to_ER} and Lemma \ref{lem:compose_ER_derivations}.

Suppose, for example, that $\Sat(a,x)$ is given as a straightforward procedure that considers each clause of $a$ in succession 
and checks for each literal whether it belongs to the clause and whether it is satisfied under $x$; 
afterwards, it checks if at least one literal was satisfied in each clause.
Let $m$ be the number of clauses in $A$. 
Then we may assume that $\norm{\Sat(a,x)}_{\restrict \tau}$ contains the following auxiliary variables:
for each $j = 1, \ldots, m$ and each $i = 1,\ldots, 2n$ (this being the number of literals in $n$ variables),
a variable $z^j_i$ with the intuitive meaning that one of the first $i$ literals appears in the $j$-th clause and happens to be satisfied; and for each $j = 1, \ldots, m$, a variable $w_j$ with the intuitive meaning that each of the first $j$ clauses 
contains a satisfied literal. We may also assume that $\norm{\Sat(a,x)}_{\restrict \tau}$ consists of the following conjuncts:
\begin{itemize}
\item for $j= 1, \ldots, m$ and $i = 1,\ldots, 2n$, a constant-size set of clauses logically equivalent to
either $z^j_{i} \leftrightarrow z^j_{i-1}$, 
if the $i$-th literal does not appear in the $j$-th clause, 
or $z^j_{i} \leftrightarrow z^j_{i-1} \vee (\neg) x_k$, 
if the $i$-th literal appears in the $j$-th clause and happens to be $(\neg)x_k$ 
(here $z^j_0$ is the constant $0$),
\item for $j =1, \ldots, m$, a constant-size set of clauses logically equivalent to with
$w_j \leftrightarrow w_{j-1} \wedge z^j_{2n}$ 
(here $w_0$ is the constant $1$),
\item $w_m$.
\end{itemize}
The derivation of this from $A$ is the obvious one: first, for each $j = 1,\ldots, m$,
introduce each $z^j_{i} \leftrightarrow z^j_{i-1}$ resp.~$z^j_{i} \leftrightarrow z^j_{i-1} \vee (\neg) x_k$ 
as extension axioms and perform a series of resolutions with the $j$-th clause of $A$, eventually obtaining $z^j_{2n}$.
Then, again for each $j = 1,\ldots, m$, introduce $w_j \leftarrow w_{j-1} \wedge z^j_{2n}$ as an extension axiom
and resolve $w_j \leftarrow w_{j-1} \wedge z^j_{2n}$ with $z^j_{2n}$ and $w_{j-1}$ to obtain $w_j$.
\end{proof}

\subsection{Simulating $\ERpls$ by linear dominance (from Section~\ref{sec:lin-dom>er-pls})}
 \label{appendix_simulate_ER}

We give some details omitted from Section~\ref{sec:lin-dom>er-pls}.
We first show how to simulate the $\ER$ rule of $\ERpls$ in the linear dominance system,
and then we fill a gap in how we handled the dominance rule.

\paragraph{ER rule.} 
Suppose $\Gamma \wedge \Delta$ is derived from $\Gamma$ by the ER rule.
Then there is an ER derivation~$\pi$ which begins with $\Gamma$ and includes every
clause in $\Delta$. 
It is enough to show that from $(\Gamma^*, \emptyset, \top, \emptyset)$
we can derive $(\Gamma^*, \pi^*, \top, \emptyset)$, where we write $\pi^*$
for the set of translations of clauses appearing in $\pi$. This is because we can then 
copy all of $\Delta^*$ from the derived constraints~$\pi^*$ into the core constraints~$\Gamma^*$
using the transfer rule of the dominance system, to get $((\Gamma \wedge \Delta)^*, \pi^*, \top, \emptyset)$.
and finally reset the derived constraints to $\emptyset$ using the deletion rule,
to get $((\Gamma \wedge \Delta)^*, \emptyset, \top, \emptyset)$ as required.

Treating each clause in $\pi$ in turn, either it is an initial clause from $\Gamma$;
or it was derived from earlier clauses by resolution or weakening;
or it is a clause introduced by the extension rule. The first two cases are
easily dealt with by the implicational derivation rule (of the dominance system),
plus the well-known simulation of resolution by CP~\cite{cook1987complexity}.
The remaining case is the extension rule.
We use the standard arguments for showing that we can simulate
 this rule in a system based on adding certain redundant clauses, see e.g.~\cite{kullmann1999generalization}.

An extension axiom in~$\pi$ has the form of three clauses $\neg u \vee \neg v \vee y$,
$\neg y \vee u$ and $\neg y \vee v$,
which we translate into three PB constraints:
\begin{align*}
A: & \quad (1-u) + (1-v) + y \ge 1\\
B: & \quad (1-y) + u \ge 1\\
C: & \quad (1-y) + v \ge 1.
\end{align*}
We show that these can be derived from any set $\EE$ of PB constraints 
which do not mention the variable $y$,
by three applications of the redundance-based strengthening rule of the dominance system 
(we may ignore the part of the rule having to do with the ordering, as we are using the trivial ordering~$\top$).
That is, we must find substitutions $\sigma, \tau, \omega$ and 
derivations 
$\EE \cup \{ \neg A \} \vCP ( \EE \cup \{ A \} )_{\restrict \sigma}$,
$\EE \cup \{ A, \neg B \} \vCP (\EE \cup \{A, B\})_{\restrict \tau}$
and
$\EE \cup \{ A, B, \neg C \} \vCP (\EE \cup \{A, B, C\})_{\restrict \omega}$.

Let $\sigma$ map $y \mapsto u$ and do nothing else.
Then $\EE_{\restrict \sigma} = \EE$ and $A_{\restrict \sigma}$ is $2-v \ge 1$,
which is the Boolean axiom $v \le 1$. Thus we have the first CP derivation.
For the second CP derivation we can set $\tau = \sigma$ again,
since $B_{\restrict \sigma}$ is just $1 \ge 1$.


Now let $\omega$ map $y \mapsto 0$ and do nothing else.
Then $\EE_{\restrict \omega} = \EE$ and we have
\begin{align*}
A_{\restrict \omega} &\textrm{ is } (1-u)+(1-v) \ge 1  
	&&  \! C_{\restrict \omega} \textrm{ is } v \ge 0  \\
B_{\restrict \omega} & \textrm{ is } u \ge 0 
	&& \neg C \textrm{ is } 1-y +v \le 0
\end{align*}
Thus $B_{\restrict \omega}$ and $C_{\restrict \omega}$ are trivially derivable.
For $A_{\restrict\omega}$,
we can rearrange it as $u + v \le 1$, which we can derive 
by starting  with $\neg C$ and adding axioms $y \le 1$ and $u \le 1$.
Thus we have the third CP derivation.


We also allow extension axioms expressing $y \leftrightarrow u$ or
$y \leftrightarrow 0$ or $y \leftrightarrow 1$. We can derive the $^*$ translations
of these in a similar way to the first two derivations above, by simply setting $y$
to be $u$ or the desired value.

\paragraph{Dominance rule.}
We fill in a step that was omitted in Section~\ref{sec:lin-dom>er-pls}.
Suppose we are in a configuration
$(\Gamma^*, \emptyset, \OOl, \vec x)$,
where $\OOl$ is the lexicographic order.
We must derive $(\Gamma^*, \Delta^*, \OOl, \vec x)$,
where $\Delta$ is a set of extension axioms over $\vec x ; \vec y$.

We do this by adding each 
extension axiom in $\Delta$ in turn using the redundance-based strengthening rule, 
in the same way that we handled introducing 
extension axioms in the case of the ER rule above. 
However this time we must check that we satisfy the order condition for this rule
-- this did not matter before, as there we had the trivial order.

For example (using the same notation $\EE$, $A$, $\sigma$ as above), in the derivation
needed to introduce $A$, the formal requirement is to show
\[
\EE \cup \{ \neg A \} \vCP ( \EE \cup \{ A \} )_{\restrict \sigma}
 \cup \OOl(\vec x_{\restrict \sigma}, \vec x).
\]
However this is trivial as 
$\sigma$ does not affect the variables $\vec x$, which are explicitly
the only variables compared in the ordering $\OOl$. The same goes for the other substitutions used
above.

\end{document}